\documentclass[copyright,creativecommons]{eptcs}

\providecommand{\event}{LTT 2026}

\usepackage{iftex}
\ifpdf
  \usepackage{underscore}         
  \usepackage[T1]{fontenc}        
\else
  \usepackage{breakurl}           
\fi

\usepackage{graphicx}
\usepackage{ebproof}
\usepackage{amssymb}
\usepackage{amsthm}
\usepackage{amsmath}
\usepackage{hyperref}
\usepackage{cleveref}
\usepackage{cite}
\usepackage{wrapfig}
\usepackage{MnSymbol}
\usepackage{array}\newcounter{rowVG}

\usepackage{tikz}\usepackage{tikzpagenodes}
\usetikzlibrary{arrows,shapes.geometric,shapes.misc,positioning,automata}
\usetikzlibrary{trees,decorations.pathmorphing,calc}
\usetikzlibrary{intersections}

\usepackage[most]{tcolorbox}
\tcbset{on line, 
        boxsep=1pt, left=1pt,right=0pt,top=0pt,bottom=0pt,
        colframe=white,colback=gray!30,  
        highlight math style={enhanced}
      }
\colorlet{lastCLR}{yellow} 
\colorlet{occursCLR}{lime} 
\colorlet{intervalCLR}{green} 
\colorlet{NOTlastCLR}{teal} 
\colorlet{ANDlastCLR}{magenta!80} 

 \usepackage{listings}
\newcommand{\clrBACK}[1]{\tikz[remember picture,overlay]{\draw[#1,line width=10.99pt,opacity=0.3] let \p1 = (current page text area.west) in (\x1,.11) -- (0,.11);}}
\newcommand{\clrLINE}[1]{\tikz[remember picture,overlay]{ }}

\newcommand{\FDadd}[1]{#1}

\newcommand{\FDaddCR}[1]{#1}
\usepackage{todonotes}
\usepackage[normalem]{ulem}

\newtheorem{thm}{Theorem}[section]
\newtheorem{lemma}[thm]{Lemma}

\newtheorem{fact}[thm]{Fact}
\theoremstyle{definition}
\newtheorem{defn}{Definition}[section]

\newtheorem{exmp}[defn]{Example}
 \newtheorem{remark}[defn]{Remark}

\newtheorem{notation}[defn]{Notation}

\newcommand{\NULL}{\ensuremath{\mathtt{NULL}}}
\newcommand{\MALLOC}{\ensuremath{\mathtt{MALLOC}}}
\newcommand{\MFREE}{\ensuremath{\mathtt{MFREE}}}

\newcommand{\FT}{\ifmmode\textsf{FT} \else \textsf{AF}\ \fi}

\newcommand{\CB}{\texttt{CB}}
\newcommand{\CK}{\ifmmode\textsc{at} \else \textsc{at}\ \fi}

\newcommand{\typeof}{\ensuremath{\mathsf{type}}}
\newcommand{\lenghtof}{\ensuremath{\#}}
\newcommand{\maxtype}{\ensuremath{\mathsf{max}_\le}}
\newcommand{\prototypeof}{\ensuremath{\mathsf{prototype}}}
\newcommand{\domof}{\ensuremath{\mathsf{dom}}}
\newcommand{\extDom}{\mathop{\mathsf{dom}\mkern-21mu\overset{int}{\phantom{=}}\mkern9mu}}

 \newcommand{\uop}{\ifmmode\textsf{uop}\else\textsf{uop}\ \fi}
 \newcommand{\bop}{\ifmmode\textsf{bop}\else\textsf{bop}\ \fi}

\newcommand{\deriv}{\vdash}

\title{A Core Calculus for Type-safe Product Lines of C Programs}

\author{Ferruccio Damiani
\institute{
Department of Computer Science (DI) 
\\
University of Turin, Turin, Italy}
\email{ferruccio.damiani@unito.it}
\and
Daisuke Kimura
\institute{
Toho University
\\
Chiba, Japan}
\email{kmr@is.sci.toho-u.ac.jp}
\and
Luca Paolini
\institute{
Department of Computer Science (DI) 
\\
University of Turin, Turin, Italy}
\email{luca.paolini@unito.it}
\and
Makoto Tatsuta
\institute{
National Institute of Informatics 
\\
Tokyo, Japan}
\email{tatsuta@nii.ac.jp}
}

\def\titlerunning{A Core Calculus for Type Safe Product Lines of C Programs}
\def\authorrunning{F. Damiani et al.}
\begin{document}
\maketitle

\begin{abstract}
  In this paper we: (1) propose Lightweight C (LC), namely a core calculus that formalizes a proper subset of the ANSI C without preprocessor directives; (2)   define Colored LC (CLC),
 namely LC endowed with ANSI C preprocessor directives; and (3) define a 
  type system for CLC that guarantees that all programs  to be generated by the C preprocessor  are well-typed C programs.
We believe that the simple formalization provided by CLC could be useful also for teaching purposes.

Stefano Berardi spent most of his academic career at the Department of Computer Science of the University of Turin, where he  conducts outstanding research on the logical foundations of computer science and on type-based program analyses. Over the years, he taught many courses, \FDaddCR{from BSc courses on programming with C
to PhD courses on program analysis.} Therefore, this paper fully falls within Stefano Berardi's  research and teaching activities.
\end{abstract}

\begin{flushright}
    \footnotesize To Stefano Berardi on the occasion of his $2^6$th birthday.
\end{flushright}

\section{Introduction}\label{sec:introduction}
A \emph{Software Product Line (SPL)}
is a family of similar programs, called \emph{variants},  with  well-documented commonalities and differences, generated
from a common code base~\cite{Clements:2001,Pohl:2005,DBLP:books/daglib/0032924}. 
Each variant of an SPL can be identified by a set of \emph{features}, called a \emph{product}, where a feature is a name associated to a program functionality.

In the annotative approach to SPL implementation, some code fragments of the code base are associated with propositional formulas over features. Then, the
variant associated to a given product is generated by removing from the code base the code fragments that are annotated by a propositional formula that evaluates to false when the features in the product assume value true and the remaining features assume value false. The \#define and \#if directives of the C preprocessor \cite{book1988ansiC} are widely used to write SPLs of C programs according to the annotative approach. 

\paragraph*{Contribution of the Paper}
The contribution of this paper is twofold:
\begin{enumerate}
\item
we propose \emph{Lightweight C (LC)}, a core calculus for C programs which models syntax and  typing 
of a subset of C; and
\item
we propose \emph{Colored LC (CLC)}, a core calculus for annotation-based SPLs of LC programs together with a family-based 
type checking  system which guarantees that 
all the variants of a well-typed CLC SPL are well typed LC programs.
\end{enumerate}
Although toolchains including a type checker for C code with preprocessor directives are available in the \FDaddCR{literature~\cite{typechef:FOSD-2010,KOE:OOPSLA12},}
we are not aware of any formalization of type checking of SPLs for a core calculus that is a proper subset of C. 
We believe that the simple formalization provided by CLC could be useful also for teaching purposes.

Stefano Berardi spent most of his academic career at the Department of Computer Science of the University of Turin, where he  conducts outstanding research on the logical foundations of computer science and on type-based program analyses, inspiring and motivating students and young researchers (including -- in  chronological order -- the first author and the third author of this paper) and collaborated with several researchers (including  -- in reverse chronological order -- the second, the fourth and the first author of this paper). Over the years, he taught many courses, \FDaddCR{from BSc courses on programming with C
to PhD courses on program analysis.} Therefore, this paper fully falls within Stefano Berardi's  research and teaching activities.

\paragraph*{Organization of the Paper} \Cref{sec:background} introduces some background and discusses some related work on SPLs.  \Cref{sec:running-example} illustrates the annotative approach to the implementation of SPLs of C programs by an example, that will be used as a running example through the paper; \Cref{sec:language} presents the syntax and the typing of LC; \Cref{sec:SPL-language} presents CLC and a family-based type system for CLC SPLs;
\Cref{sec:SPL-type-safety} shows that 
the 
type checking system guarantees that all 
 all the variants of a well-typed CLC SPL are well typed LC programs; and
 \Cref{sec:conclusion} outlines some possible directions for future work.


\section{Background and Related Work}\label{sec:background}
\subsection{Background}

An SPL can be structured into: (i) a
 \emph{feature model (FM)} describing the variants in
terms of \emph{features}, where each feature is a name representing  an abstract description of functionality and each variant is identified by a set of features, called a \emph{product}; 
(ii) a    \emph{code base (CB)}  comprising language dependent 
reusable code artifacts  that are used to build the variants; and  (iii) \emph{configuration knowledge (CK)}  connecting  feature model and  artifact base by specifying
how, given a product,  the corresponding  variant can be derived from the code artifacts---thus inducing a mapping from products to variants, called the \emph{variant generator (VG)} of the SPL.

Several approach to representation of FM{s} 
and to the implementation of  SPLs (see, e.g.,~\cite{Batory:2005,Schaefer-EtAl:STTT-2012})  have been proposed in the literature. In this paper we consider the \emph{propositional representation} of FM{s} and  the \emph{annotative approach} to SPL implementation. 

\begin{defn}[FM in propositional representation]\label{def:FM}
An FM $\Phi$
is  a pair $(\mathcal{F},\phi)$ 
where
 $\mathcal{F}$ is a (finite) 
 set of features and
 $\phi$ is a
propositional formula over $\mathcal{F}$ (i.e., whose variables 
$F$ are elements of $\mathcal{F}$), with the following syntax:
\[
\phi \;\, ::= \;\, \lstinline{F} \; \mid \;\,  !\,\phi \; \mid \; \phi\,\&\&\,\phi  \; \mid \; \phi\,|| \phi\,\; \mid \; 0 \; \mid \; 1
\]

The products of an FM are the set of features
$p\subseteq\mathcal{F}$ such that $\phi$ is
satisfied by assigning value 1 to the variables $F$ in $p$
and 0  to the variables in $\mathcal{F}\setminus p$.
\end{defn}

\begin{remark}[On the syntax used for the proposition formulas]\label{rem:LogicalOperators}
In \Cref{def:FM} we have introduced  propositional formulas over features by using logical operators and constants in the syntax of the C programming language, which is also the
 syntax used in the conditions in the \lstinline?#if? preprocessor directives. Namely: $\&\&$ for conjunction, 
  $||$ for disjunction, $!$ for negation, 0 for false, and 1 for true. 
\end{remark}

\begin{notation}[A couple of abbreviation for propositional formulas]\label{not:Rightarrow-Leftrightarrow}
Sometimes, to improved readability, we will use the following abbreviations:
  \begin{itemize}
  \item
  $\phi_1\Rightarrow \phi_2$ as short for $!\phi_1\,||\, \phi_2$; and
    \item
  $\phi_1\Leftrightarrow \phi_2$ as short for $(\phi_1 \Rightarrow \phi_2) \,\&\&\, (\phi_2\Rightarrow \phi_1)$.
  \end{itemize}
\end{notation}

According to the annotative approach to SPL implementation:
\begin{itemize}
\item
 the CB{}  is a single program (possibly divided into different files) that contains (at least) all the code that occurs in any variant;
 \item
 the CK{}  is a mapping that annotates fragments of code in the CB with propositional expressions over features; and
 \item
 the VG takes a product $p$ and removes from the CB all the fragments that are annotated with a feature expression that evaluates to 0 when assigning value 1 to the variables $F$ in $p$
and 0  to the variables in $\mathcal{F}\setminus p$.
\end{itemize}

The pre-processor of C program language, with its directives \lstinline?#define? and \lstinline?#if?, is a prominent incarnation of the annotative approach to SPL implementation.
Many well-know software distributions  \cite{freshmeat}  
(e.g., 
apache, 
berkeley db,
linux, 
php, 
python,
sqlite,
vim,
etc.)
 can be understood as annotative implementations of an SPLs of C programs~\cite{liebig-et-al:ICSE-2010}. These, SPLs often have a huge number of products.
 For instance, the Gentoo~\cite{Gentoo} source-based Linux
distribution consists of many highly-configurable packages;
the March 1st 2019 version of the Gentoo distribution comprises
671617 features spread across 36197 feature models, and  thus up to $2^{671617}$ products~\cite{Lienhardt-e-al:ICSE-2020}. A number of questions, like the following one, naturally arise:
\begin{center}
\emph{How can we ensure that the (up to)  $2^{671617}$ variants are well-typed C programs?}
\end{center}
Clearly, even for SPLs with fewer features (e.g., one FM with 64 features, thus up to $2^{64}$ products) generating and type checking each variant in isolation is unfeasible. 

Several approaches to analysis of SPLs have been proposed in the literature (surveyed, e.g., in~\cite{Thum-EtAl:ACM-CS-2014}). Among these approaches, \emph{family-based approaches} can be characterized as follows (cf.\ Definition 4.1 in \cite{Thum-EtAl:ACM-CS-2014}):
\begin{itemize}
\item[]
An analysis of an SPL is family based if it: (a) operates only on the CB; and (b) incorporates the knowledge
in the FM and CK.
\end{itemize}
The key idea of family-based approaches  is to analyze the code base and the feature model together, to check whether
some intended properties hold for all variants, aiming at the feasibility of implementing the analyses.

\subsection{Related Work}\label{sec:related}

We focus this discussion of related research on two  aspects of our work: 
core languages for C  and core calculi for the annotative approach to SPL implementation. 

\paragraph{On Core Languages for C}  Albeit endowed with few keywords, C is a rather subtle language that few programmers know in detail because their wide room for undefined behavior \cite{hathhorn2015acm}.
   A complete formalization of the semantics for the \textbf{C11} language has been provided  by Chucky Ellison in his Ph.D. thesis \cite{ellison12phd} (see also~\cite{ellison2012acm});
  the thesis surveys many attempts to provide a formal semantics  for~C. 
  
  A number of core languages for C have been proposed in literature, e.g.: Featherweight C~\cite{FeatherweightC:2012}, a  ``simple C-like language''~\cite{jeehoon2015acm}, and Clight~\cite{blazy2009jar}.
We designed LC aiming for a calculus that has a relationship to C somewhat similar to that of lambda-calculus~\cite{Church:1932} (see also~\cite{Barendregt:1985}) or PCF~\cite{Plotkin:1977} with programming languages such as ML
or Haskell, and that of Featherweight Java (FJ)~\cite{FJ:TOPLAS-2001} with Java. None of the core languages for C mentioned above satisfy this requirement: Featherweight C~\cite{FeatherweightC:2012} and the ``simple C-like language'' are on the one hand too small (they do not model structures) and on the other hand too large (they model functions' local variables and statements); while Clight is definitely too big.

It is worth observing that  there is a substantial difference between the relations: (1) of the lambda-calculus or PCF with ML
or Haskell;  and (2) of FJ with Java. Namely, the lambda-calculus and PCF are Turing complete, while FJ has not
beeind designed with the goal of being (and, to the best of our knowledge, it is not) Turing complete. 
LC is (assuming integers of arbitrary size) Turing-complete and it
goes (by including integers, \FDadd{imperative assignment of values to struct-members,} conditional-expressions, expression sequences,  operators, \FDadd{and dynamic memory allocation/deallocation on the heap)} a bit --- we believe not too much --- beyond the spirit of FJ, hence the name we choose.

\paragraph{On Core Calculi for the Annotative Approach to SPL Implementation} 
Although:
\begin{itemize}
\item
a formal account of family-based type checking of annotation-based SPLs of Java programs has been proposed  
by means of the minimal core calculus Colored Featherweight Java (CFJ)~\cite{kastner2021acm}, which formalizes family-based type checking of SPLs of 
(cast-free) Featherweight Java (FJ)~\cite{FJ:TOPLAS-2001} programs; and
\item
TyeChef,\footnote{\url{https://github.com/ckaestne/TypeChef}} an infrastructure which includes a type checker for SPLs of C programs that implements variability with the \lstinline?#define? and \lstinline?#ifdef? directives of the
C preprocessor, has been proposed~\cite{typechef:FOSD-2010}, formalized~\cite{KOE:OOPSLA12} on top of a core calculus for module systems~\cite{Cardelli:POPL-1997},  and used as basis for implementing a number of SPL analyses (see, e.g.,~\cite{LvKADL:ESECFSE13,Medeiros-et-al:GPCE-2013,vLJKA:TOSEM18});
\end{itemize}
we are not aware of any formalization of  family-based type checking for annotation-based SPLs on top of a core language that is a proper subset of C.
Our formalization of the annotative approach to the implementation of SPL of C programs is inspired by the formalization 
introduced in Kastner et al.'s  paper on CFJ~\cite{kastner2021acm}. 
We refer to that paper for a discussion of related work on core calculi for the annotative approach to SPL implementation, and to a couple of survey papers for  broader  discussions about SPL implementation approaches~\cite{Schaefer-EtAl:STTT-2012} and  analysis strategies for SPLs~\cite{Thum-EtAl:ACM-CS-2014}.
A key difference between CLC and CFJ is that LC models \FDaddCR{imperative variable (struct-member) assignment} and dynamic memory allocation/deallocation on the heap; therefore CLC provides a suitable basis for investigating family-based analyses of SPLs of C programs for properties like, e.g., heap-safety.


\section{Running Example}\label{sec:running-example}
As a running example we consider an SPL of FIFO queues of integers, implemented by  linked lists, that we call the \emph{Queue SPL}. 
All variants of the Queue SPL define a struct of name \lstinline?fifo?, which has a member \lstinline?first? of type \lstinline?struct node*?
containing the (pointer to the first node of the) list containing the elements of the queue (\lstinline?NULL? if the queue is empty), 
\FDadd{and some functions.}

The Queue SPL has  three features and six products (and, therefore, six variants). The features are: 
\begin{itemize}
\item
\lstinline?LAST?, specifying that the struct of name \lstinline?fifo? 
has a member \lstinline?last? of type \lstinline?struct node*?
containing the pointer to the last element of the list containing the elements of the queue (or \lstinline?NULL?, if the queue is empty); 
\item
\lstinline?OCCURS?, specifying that there is a function \lstinline?int occurs(struct node* list, int element)? that checks whether the given \lstinline?list? contains an occurrence of the given \lstinline?element?; and
\item
\lstinline?INTERVAL?, specifying that there is a function \lstinline?int occurs(struct node* list, int start, int end, int element)? that checks whether the given \lstinline?list? contains the given element in a position between the given positions \lstinline?start? and \lstinline?end? (the fist element of the list is position 0).
\end{itemize}
The six variants (C programs) of the Queue SPL are identified by the six products described by the FM 
$\Phi_{\mathit{queue}}=({\mathcal F}_{\mathit{queue}},\phi_{\mathit{queue}})$, where:
\begin{itemize}
\item
${\mathcal F}_{\mathit{queue}} = \{\lstinline?LAST?, \lstinline?OCCURS?, \lstinline?INTERVAL?\}$, and
\item
$\phi_{\mathit{queue}} \; = \; !\,\lstinline?INTERVAL? \; || \; \lstinline?OCCURS?$;
\end{itemize}
which are the combinations of the features in  ${\mathcal F}_{\mathit{queue}}$ that satisfy  the propositional formula 
$\phi_{\mathit{queue}}$. Namely, they are:
\begin{enumerate}
\item
$p_1 = \emptyset$,  the empty set of features, which identifies the variant in \Cref{fig:queue-var1};
\item
$p_2 = \{\lstinline?LAST?\}$, which identifies the variant in \Cref{fig:queue-var2};
\item
$p_3 = \{\lstinline?OCCURS?\}$, which identifies the variant obtained from the code in \Cref{fig:queue-var3-4} by replacing the comment with the code of Variant 1;
\item
$p_4 = \{\lstinline?LAST?, \lstinline?OCCURS?\}$, which identifies the variant obtained from the code in \Cref{fig:queue-var3-4} by replacing the comment with the code of Variant 2;
\item
$p_5 = \{\lstinline?OCCURS?, \lstinline?INTERVAL?\}$, which identifies the variant obtained from the code in \Cref{fig:queue-var5-6} by replacing the comment with the code of Variant 1; and
\item
$p_6 = \{\lstinline?LAST?, \lstinline?OCCURS?, \lstinline?INTERVAL?\}$, which identifies the variant obtained from the code in \Cref{fig:queue-var5-6} by replacing the comment with the code of Variant 2.
\end{enumerate}

\begin{figure}[t]
\begin{footnotesize}
\begin{lstlisting}
struct node {
  int data;
  struct node* next;
};
struct fifo {
   struct node* first;
};
int front(struct fifo* queue) {return queue->first->data;}
struct node* getlast(struct node* list)
 {return !(list->next) ? list : getlast(list->next);}  
int rear(struct fifo* list) {return getlast(queue->first)->data;}
int main() {return 0;} 
\end{lstlisting}
\end{footnotesize}
\vspace{-30pt}
\caption{Variant 1 (for product $p_1 = \emptyset$) of the Queue SPL.}\label{fig:queue-var1}
\end{figure}
\begin{figure}[t]
\begin{footnotesize}
\begin{lstlisting}
struct node {
    int data;
    struct node* next;  
};
struct fifo {
    struct node* first;
    struct node* last;
};
int front(struct fifo* queue) { return queue->first->data; }
int rear(struct fifo* queue) { return queue->last->data; }
int main() {return 0;}
\end{lstlisting}
\end{footnotesize}\vspace{-30pt}
\caption{Variant 2 (for product $p_2 = \{\lstinline?LAST?\}$) of the Queue SPL.}\label{fig:queue-var2}
\end{figure}
\begin{figure}[t]
\begin{footnotesize}
\begin{lstlisting}
// Insert the code of: Variant 1 to get Variant 3; and of Variant 2 to get Variant 4

int occurs(struct node* list, int element)
 {return !(list) && ((list->data == element)  || occurs(list->next,element));} 
\end{lstlisting}
\end{footnotesize}\vspace{-30pt}
\caption{Variants 3 (for $p_3 = \{\lstinline?OCCURS?\}$) and  4 (for $p_4 = \{\lstinline?LAST?, \lstinline?OCCURS?\}$) of the Queue SPL.}\label{fig:queue-var3-4}
\end{figure}
%
%
\begin{figure}[t]
\begin{footnotesize}
\begin{lstlisting}
// Insert the code of: Variant 1 to get Variant 5; and of Variant 2 to get Variant 6

 int occurs(struct node* list,
            int start, int end,  
            int element) 
 {return !(list) 
       && end >= 0 
       && ((start <= 0 && list->data == element) || occurs(list->next,start-1,end-1,element)) 
 ;}    
\end{lstlisting}
\end{footnotesize}\vspace{-30pt}
\caption{Variants 5 (for $p_5 = \{\lstinline?OCCURS?, \lstinline?INTERVAL?\}$) and  6 ($p_6 = \{\lstinline?LAST?, \lstinline?OCCURS?, \lstinline?INTERVAL?\}$) of the Queue SPL.}\label{fig:queue-var5-6}
\end{figure}


The code base of the Queue SPL is given in \Cref{fig:queue-spl}, where for each feature there is a \lstinline?#define? preprocessor directive, and the variant identified 
by a given product can be generated by  commenting out the  \lstinline?#define? preprocessor directives for the features that are not included in the product.
It is worth observing that code base contains some code (namely, the comma at line 26 in \Cref{fig:queue-spl}) that is not included in any variant; this is done to ensure that 
the code obtained by removing all preprocessor directives from the code base is a well-formed C program. This program, which is not a variant of the Queue SPL, is given in
\Cref{fig:queue-spl-stripped}.

\begin{figure}
\include{Queue-SPL}
\caption{Code base of the Queue SPL, where the code associated to the annotations \tcbox[colback=lastCLR!30]{LAST}, \tcbox[colback=occursCLR!30]{OCCURS}, \tcbox[colback=intervalCLR!30]{INTERVAL}, \tcbox[colback=NOTlastCLR!30]{!LAST}, and \tcbox[colback=ANDlastCLR!30]{LAST \&\& !LAST} (i.e., false) is highlighted with different colors.}\label{fig:queue-spl}
\end{figure}

\begin{figure}
\include{Queue-SPL-stripped}
\caption{C program obtained by removing the preprocessor directives from the Queue SPL code base, where the code associated to the annotations \tcbox[colback=lastCLR!30]{LAST}, \tcbox[colback=occursCLR!30]{OCCURS}, \tcbox[colback=intervalCLR!30]{INTERVAL}, \tcbox[colback=NOTlastCLR!30]{!LAST}, and \tcbox[colback=ANDlastCLR!30]{LAST \&\& !LAST} (i.e., false) is highlighted with different colors.}\label{fig:queue-spl-stripped}
\end{figure}

\begin{remark}[A couple of  annotative C SPL programming patterns]\label{rem:SPL-patterns}
The code in lines 22-30 of~\Cref{fig:queue-spl} is an instance of the following annotative C SPL programming pattern. 
\begin{description}
\item[Pattern 1]
To provide $n\ge 2$ 
alternative versions of an expression (to be included in the code of different variants),  \FDaddCR{proceed as follows. First, provide the sequence $(e_1,...,e_n)$ of the $n$ 
expressions. Then, annotate the comma between $e_j$ and $e_{j+1}$ $(1\le j\le n-1)$ with a propositional formula that evaluates always to 0 (e.g., by 0 itself). Finally, 
annotate each expression $e_i$ $(1\le i \le n)$ with a suitable propositional formula $\phi_i$ over features; i.e., the formulas   $\phi_i$ must be}
(modulo the FM) pairwise mutually exclusive and (modulo the FM) at least one of them must be always satisfied. 
\end{description}
Instead of annotating the comma in line 26 with ``0'', we have  annotated it with the formula ``LAST \&\& !LAST''. 
This makes the code in lines 22-30 also an instance of the following more general annotative C SPL programming pattern:
\begin{description}
\item[Pattern 2]
\FDaddCR{Consider a  sequence $(e_1,...,e_n)$ of  $n\ge 2$ 
expressions. Annotate each expression $e_i$ $(1\le i \le n)$  with a propositional formula $\phi_i$ over features
such that (modulo the FM) at least one of the $\phi_i$' is always satisfied.
Annotate  the $n-1$ commas separating the $n$ expressions in the sequence  $(e_1,...,e_n)$ with the $n-1$ formulas $\psi_i$ $(1\le i\le n-1)$ (respectively), where $\psi_i$ is  defined as: 
  $$\psi_i = \phi_i\;\&\&\;(||_{i<j<n} \phi_j) = \phi_i\;\&\&\;(\phi_{i+1}\;||\cdots||\; \phi_{n-1})\;.$$ 
This ensures that, whenever $m$ expressions $(m<n)$ of the $e_i$ $(1\leq i\leq n)$ are removed,  then $m$ commas are removed in such a way the 
remaining code forms a well-formed sequence of expressions.} 
\end{description}
\end{remark}


%
%
%
%
%
%
%
%


\section{Lightweight C}\label{sec:language}

In this section we  introduce  \emph{Lightweight C  (LC)}, a core calculus for C programs. 

\subsection{LC Syntax}

The  abstract syntax of LC is given in \Cref{fig:LC-syntax} where, following~\cite{FJ:TOPLAS-2001,pierce2002book}, we adopt the overline notation  $\overline{\cdot}$
to denote  (possibly empty) sequences. We use the unary function $\lenghtof(\cdot)$  do denote the length of a sequence.
In particular, ``$\overline{\lstinline!SD!}$'', ``$\overline{\lstinline!T!} \, \overline{\lstinline!m!};$'',  ``$\overline{\lstinline?T?} \, \overline{\lstinline?x?}$'',  
and ``$\overline{\lstinline!e!}$'' denote
a sequence of struct definitions ``$\lstinline!SD!_1\cdots\lstinline!SD!_n$'' for $n=\lenghtof(\overline{\lstinline!SD!})$,
 a sequence of  struct members declarations ``$\lstinline!T!_1\,\lstinline!m!_1;\ldots\lstinline!T!_n\,\lstinline!m!_n;$'' 
 for $n=\lenghtof(\overline{\lstinline!T!} \, \overline{\lstinline!m!};)$, 
 a sequence of  function formal parameter declarations ``$\lstinline!T!_1\, \lstinline!x!_1,\ldots,\lstinline!T!_n\,\lstinline!m!_n$'' 
 for $n=\lenghtof(\overline{\lstinline?T?} \, \overline{\lstinline?x?})$, and
 a sequence of expressions ``$\lstinline!e!_1,\ldots,\lstinline!e!_n$'' for $n=\lenghtof(\overline{\lstinline!e!})$, respectively. The empty sequence is denoted by ``$\bullet$'' and a non-empty sequence of expressions is denoted $\tilde{\lstinline!e!}$.
We assume that the elements of list of named entities (i.e., struct definitions, struct member declarations, function definitions, function formal parameter declarations) have different names.

 The expression ``\MALLOC(\lstinline?struct?\lstinline?s?)'' \FDadd{is} the abbreviation of 
``\lstinline?(struct?\lstinline?s?$\ast$\lstinline?) malloc (sizeof(struct?\lstinline?s?\lstinline?)))?''
where \lstinline?malloc? is the standard library function for memory allocation;\footnote{We might assume that every LC program implicitly begins with the directive that define the parametric macro: ``\lstinline?\#define ?\MALLOC\lstinline?(a) (a?$\ast$\lstinline?)malloc(sizeof(a))?''.} and \MFREE\ \FDadd{is}  a variant of the \lstinline?free? standard library function for memory allocation
 which has return type ``\lstinline?void?$\ast$'' (instead of ``\lstinline?void?'') and always returns \NULL. 
 
\begin{remark}[On the syntax of LC and the syntax of ANSI C]\label{rem:LC-syntax}
The syntax in \Cref{fig:LC-syntax} \FDadd{is} subset of 
 \FDadd{the syntax defined} in the C24 Standard. In particular, we note that in accord to Sect.6.5.18 (Comma Operator)  of  \cite{C24-standard}  non-empty parenthesized expression-sequences ($\widetilde{\lstinline!e!}$) can be supplied as a single argument in a function \FDadd{call; e.g., given a funtion with prototype ``$\lstinline!int f(int x);!$''  we can 
 call it as
 ``$\lstinline!f((2,4))!$'', which is equivalent to ``$\lstinline!f(4)!$''.}
 We remind that references to struct-names are allowed without requiring that their definitions be available before\footnote{This is possible because,
   in \cite[Sect.6.2.5]{C24-standard}  is explained that  ``All pointers to structure types shall have
   the same representation and alignment requirements as each other.'' while, on the other hand, different kinds of pointers may have distinct width and/or alignment properties.},
 thus allowing recursive and mutual recursive types definitions. 
\end{remark}
 
 \begin{figure}
  \begin{tabular}{lcl@{\qquad}l}
  \hline\\[-5pt]
   \lstinline!Prg! & ::=  & $\overline{ \lstinline!SD! }$ $ \overline{ \lstinline!FD!}$
    &  Program  \\
    \lstinline!SD! & ::=  &\lstinline?struct s?\{$\overline{ \lstinline!T! }\ \overline{\lstinline!m!};$\};    
    &  Struct definition  \\ 
    \lstinline!FD! & ::=  &\lstinline?T f(?$\overline{\lstinline?T?}\ \overline{\lstinline?x?}$\lstinline?) { return e; }?     &  Function definition  \\
   \lstinline!T! & ::=  &  \lstinline?int?  $\;\mid$ \lstinline?void?$\ast$ $\mid$ \lstinline?struct?\lstinline?s?$\ast$  & Type \\  
    \lstinline!e! & ::=  &   n $\;\mid \;\NULL\; \mid$ \lstinline?x?
                           $\mid$  \lstinline?f? ($\bar{\lstinline!e!}$) 
                           $\mid$ \lstinline? e->m ?                            
                           $\mid$ \lstinline?e->m=e?   
                           $\mid$ \lstinline!e?e:e!    
                           $\mid$ ($\tilde{\lstinline'e'}$)                                                 
                                    &    \\
 &  &    
                           $\mid$ \uop \lstinline?e?
                           $\mid$ \lstinline?e? \bop \lstinline? e?     
                           $\mid$  \MALLOC(\lstinline?struct?\lstinline?s?)
                           $\mid$  \MFREE(\lstinline?e?)
    &  Expression  \\
  \hline      
 \end{tabular} 
  \caption{LC abstract syntax, where \lstinline?s? ranged over struct names, \lstinline?m? over struct member names, \lstinline?f? over function names, \lstinline?x? over function formal parameter names, n over interger values, \texttt{uop} over unary operators and \texttt{bop} over binary operators.}
  \label{fig:LC-syntax}
\end{figure}

 \FDadd{A program} consists of a sequence of struct declarations followed by a sequence of function declarations which always contains a function  \lstinline?main?
 which takes no arguments and returns an integer.
 
\FDadd{To simplify the presentation,  in what follows we consider, in the obvious way,} a struct definition \lstinline!SD! as a mapping from each member name \lstinline!m!
to member declaration \lstinline!T m!, a sequence of struct definitions $\overline{\lstinline!SD!}$ as a mapping from each struct name \lstinline!s! to the struct definition \lstinline!SD!, a sequence of function definitions $\overline{\lstinline!FD!}$ as a mapping from a function name \lstinline!f! to function definition \lstinline!FD!, and a program $\lstinline!Prg!$ as the union of the mappings  $\overline{\lstinline!SD!}$ and $\overline{\lstinline!FD!}$.  

 We write ``\lstinline!Prg! \textsc{sane}'' to mean that the program $\lstinline!Prg!=\overline{\lstinline!SD!} \; \overline{\lstinline!FD!}$  satisfies the following sanity conditions:
(1)
for every struct name  \lstinline?s? appearing everywhere in \lstinline!Prg! we have  $\lstinline?s?\in\domof(\overline{\lstinline!SD!})$; 
(2)
for every function name  \lstinline?f? appearing everywhere in $\overline{\lstinline!FD!}$ we have  $\lstinline?f?\in\domof(\overline{\lstinline!FD!})$;
and
(3) $\overline{\lstinline!FD!}(\lstinline?main?) = \lstinline?int main(){return e;}?$, for some expression \lstinline?e?.


\begin{exmp}[Some LC programs]\label{exa:LC-programs}
The six programs illustrated in \Cref{fig:queue-var1,fig:queue-var2,fig:queue-var3-4,fig:queue-var5-6} are LC programs.
\end{exmp}

\subsection{LC Typing}





Unary and binary operators are listed in \Cref{fig:LC-operators}, together with their types and a brief description. Given a (unary or binary) operator \texttt{op}, we write 
\typeof(\texttt{op}) to denote its type.

\begin{figure}
  \begin{tabular}{l@{\quad}l@{\quad}l}
         \hline 
    \textbf{Name} & \textbf{Type} & \textbf{Description} \\
              \hline 
    \emph{Unary operators}  \\
       \hline
      \lstinline!-! & (\lstinline?int?)$\rightarrow$\lstinline?int?    &  unary minus operators  \\
       \lstinline?!?  & (\lstinline?T?)$\rightarrow$\lstinline?int?    &  returns whether its argument is different from 0 and \NULL\\
                \hline 
    \emph{Binary operators}  \\
       \hline         
\lstinline?+, -, *, /, %?  & (\lstinline?int?,\lstinline?int?)$\rightarrow$\lstinline?int?    &  arithmetic operators \\      
\lstinline?&&, ||?  & (\lstinline?int?,\lstinline?int?)$\rightarrow$\lstinline?int?    &  logic operators \\  
     \lstinline^<, <=, >, >=^ & (\lstinline?int?,\lstinline?int?)$\rightarrow$\lstinline?int?    & integer relational operators  \\      
  \lstinline^==, !=^ & (\lstinline!T!,\lstinline!T!)$\rightarrow$\lstinline?int?    &  equality and inequal operators  \\ 
  \hline
\end{tabular}
   \caption{LC operators.}
  \label{fig:LC-operators}
\end{figure}

Following the convention in~\cite{pierce2002book},  a \emph{type environment} $\Gamma$ is a finite mapping from function formal parameter names to types, written ``$\overline{\lstinline?x?}:\overline{\lstinline?T?}$'' which (as usual) is short for
``$\lstinline?x?_1:\lstinline?T?_1,...,\lstinline?x?_n:\lstinline?T?_n$'' with $n\ge 0$. In \Cref{fig:LC-typing}, we define three kinds of typing judgements: we write $\vdash \lstinline?Prg?  \; \textsc{ok}$ to denote that \lstinline?Prg? is well typed;  we write $\vdash \lstinline{FD} \; \textsc{ok}$ to denote that functions defined in  $\domof(\lstinline{FD})$ are well typed; and,  we write $\Gamma \vdash \lstinline{e}: \lstinline?T?$ to denote that the expression  \lstinline{e} has type \lstinline?T? under the assumptions in the environment $\Gamma$.
\FDadd{The}
 symbol $\le$ denotes the subtyping relation,
 which is the  reflexive closure of the relation 
 defined by 
\lstinline?void?$\ast$ $\le$  \lstinline?struct?\lstinline?s?$\ast$, for all struct \lstinline?s? defined in the program.
We write $\maxtype\{\lstinline?T?_1,\lstinline?T?_2\}$ to denote the greatest between types $\lstinline?T?_1$ and $\lstinline?T?_2$.

The typing rules are fairly standard. 
Following~\cite{pierce2002book}, in
the premises of rule \textsc{T-prg} we write ``$\vdash \overline{\lstinline?FD?}  \; \textsc{ok}$''
as short for ``$\vdash \lstinline?FD?_1  \; \textsc{ok}$ $\cdots$ $\vdash \lstinline?FD?_n  \; \textsc{ok}$'';
in the premises of rule
\textsc{T-app} we write ``$\Gamma \vdash \overline{\lstinline{e}} :  \overline{\lstinline?T?}\,$'' as short for 
``$\Gamma \vdash \lstinline{e}_1 :  \lstinline?T?_1$ $\cdots$ $\Gamma \vdash \lstinline{e}_n :  \lstinline?T?_n$'' and, similarly, 
for the  premise ``$\Gamma \vdash \tilde{\lstinline{e}} :  \tilde{\lstinline?T?}$'' of rule \textsc{T-seq}; while in premises of rule
\textsc{T-app} we write ``$\overline{\lstinline{T}}'\le\overline{\lstinline{T}}$'' as short for ``$\lstinline{T}'_1\le\lstinline{T}_1$ $\cdots$ $\lstinline{T}'_n\le\lstinline{T}_n$''.
We note that \textsc{T-malloc} requires $\lstinline?s?\in\domof(\overline{ \lstinline!SD!})$ albeit this is ensured by the sanity condition in the rule \textsc{T-prg}: this just simplifies some statetement about the system.

\begin{figure}
  \hrule
    \vspace{2pt}
  \textbf{Program typing} \hfill \fbox{$\vdash \lstinline?Prg?  \; \textsc{ok}$}
   \vspace{5pt}
  \\
  \begin{prooftree}[right label template=\tiny(\inserttext)]
     \hypo{ \lstinline!Prg! \; \textsc{sane} }
 \hypo{ \lstinline!Prg! = \overline{\lstinline?SD?}\;\overline{\lstinline?FD?} }
  \hypo{ \vdash \overline{\lstinline?FD?}  \; \textsc{ok}}
\infer3[\textsc{T-prg}]{ \vdash \lstinline!Prg! \; \textsc{ok} }
\end{prooftree}

\vspace{5pt}
\hrule
    \vspace{2pt}
  \textbf{Function typing} \hfill \fbox{$\vdash \lstinline{FD} \; \textsc{ok}$}
   \vspace{5pt}
  \\
  \begin{prooftree}[right label template=\tiny(\inserttext)]
  \hypo{ \overline{\lstinline?x?} : \overline{\lstinline?T?} \vdash \lstinline{e}: \lstinline{T}_1 }
    \hypo{ \lstinline?T?_1\le\lstinline?T?_0} 
\infer2[\textsc{T-fun}]{ \vdash \lstinline?T?_0 \; \lstinline?f(?\overline{\lstinline?T?} \; \overline{\lstinline?x?}\lstinline?) \{return e?\lstinline?;\}? \; \textsc{ok}}
\end{prooftree}

  \vspace{5pt}
\hrule
    \vspace{2pt}
  \textbf{Expression typing} \hfill \fbox{$\Gamma \vdash \lstinline{e}: \lstinline?T?$}
   \vspace{5pt}
  \\
\begin{prooftree}[right label template=\tiny\textsc{(\inserttext)}]
\infer0[T-int]{ \Gamma \vdash \text{n}: \lstinline{int} }
\end{prooftree}
\qquad
\begin{prooftree}[right label template=\tiny\textsc{(\inserttext)}]
\infer0[T-null]{ \Gamma \vdash \NULL: \lstinline{void}\ast }
\end{prooftree}
\qquad
\begin{prooftree}[right label template=\tiny\textsc{(\inserttext)}]
\hypo{ \lstinline?x?:\lstinline?T?\in\Gamma}
\infer1[T-par]{ \Gamma \vdash  \lstinline?x? :\lstinline?T? }
\end{prooftree}

\bigskip

\begin{prooftree}[right label template=\tiny\textsc{(\inserttext)}]
 \hypo{ \lstinline!Prg!(\lstinline?f?)=\lstinline?T?_0 \; \lstinline?f(?\overline{\lstinline?T?} \; \overline{\lstinline?x?})\{ \ldots\}  }
 \hypo{ \Gamma \vdash \overline{\lstinline{e}} :  \overline{\lstinline?T?}'}
   \hypo{ \overline{\lstinline?T?}'\le\overline{\lstinline?T?}} 
\infer3[T-app]{   \Gamma \vdash \lstinline?f? (\overline{ \lstinline!e!}) : \lstinline!T!_0   }
\end{prooftree}
\quad
\begin{prooftree}[right label template=\tiny\textsc{(\inserttext)}]
 \hypo{ \Gamma \vdash \lstinline{e}_0 :\lstinline{struct s}\ast}
 \hypo{ \lstinline!Prg!(\lstinline{s})(\lstinline{m}) =  \lstinline?T m? }
\infer2[T-member]{  \Gamma \vdash \lstinline{e}_0\lstinline{->m} : \lstinline?T? }
\end{prooftree}

\bigskip

\begin{prooftree}[right label template=\tiny\textsc{(\inserttext)}]
  \hypo{ \Gamma \vdash \lstinline{e}_0: \lstinline{struct s}\ast }
   \hypo{ \lstinline!Prg!(\lstinline{s})(\lstinline{m}) =  \lstinline?T m?}
   \hypo{ \Gamma \vdash \lstinline{e}_1:\lstinline?T?_1}
    \hypo{ \lstinline?T?_1\le\lstinline?T?} 
\infer4[T-assign]{ \Gamma \vdash \lstinline{e}_0\lstinline{->m}=\lstinline{e}_1 : \lstinline?T? }
\end{prooftree}

\bigskip

\begin{prooftree}[right label template=\tiny\textsc{(\inserttext)},center=false]
 \hypo{ \Gamma \vdash \lstinline{e}_0: \lstinline?T?_0 }
\hypo{ \Gamma \vdash \lstinline{e}_1:\lstinline?T?_1}
\hypo{ \Gamma \vdash \lstinline{e}_2:\lstinline?T?_2}
 \hypo{ \lstinline?T?_3=\maxtype\{\lstinline?T?_1,\lstinline?T?_2\} } 
\infer4[T-cond]{ \Gamma \vdash \lstinline!e!_0\lstinline!?e!_1\lstinline!:e!_2 : \lstinline!T!_3 }
\end{prooftree}
\bigskip

\begin{prooftree}[right label template=\tiny\textsc{(\inserttext)},center=false]
    \hypo{  \tilde{\lstinline!e!}=\lstinline!e!_1,\ldots,\lstinline!e!_n}
  \hypo{ \Gamma \vdash \lstinline{e}_1 :  \lstinline?T?_1 \quad\cdots\quad \Gamma \vdash \lstinline{e}_n :  \lstinline?T?_n }
    \hypo{ \lstinline?T? = \lstinline?T?_n}
\infer3[T-seq]{   \Gamma \vdash (\tilde{\lstinline{e}}) :  \lstinline?T? }
\end{prooftree}
\bigskip

\begin{prooftree}[right label template=\tiny\textsc{(\inserttext)},center=false]
 \hypo{ \Gamma \vdash \lstinline{e}_0 :\lstinline?T?_0}
 \infer[no rule]1{ \typeof(\texttt{uop})=(\lstinline?T?_0)\rightarrow \lstinline?int?  }
\infer1[T-uop]{   \Gamma \vdash \texttt{uop}(\lstinline!e!_0) :  \lstinline?int? }
\end{prooftree}
\qquad
\begin{prooftree}[right label template=\tiny\textsc{(\inserttext)},center=false]
    \hypo{  \Gamma \vdash \lstinline{e}_1  : \lstinline?T?_1}
    \infer[no rule]1{ \lstinline?T?_3=\maxtype\{\lstinline?T?_1,\lstinline?T?_2\}  }
    \hypo{  \Gamma \vdash \lstinline{e}_2 : \lstinline?T?_2}
     \infer[no rule]1{ \typeof(\texttt{bop})=(\lstinline?T?_3,\lstinline?T?_3)\rightarrow \lstinline?int?}
\infer2[T-bop]{ \Gamma \vdash \lstinline{e}_1 \;\texttt{bop}\; \lstinline{e}_2 : \lstinline?int? }
\end{prooftree}

\bigskip

\begin{prooftree}[right label template=\tiny\textsc{(\inserttext)},center=false]
  \hypo{\lstinline?s?\in\domof(\overline{ \lstinline!SD!})}
\infer1[T-malloc]{ \Gamma \vdash \MALLOC(\lstinline?struct?\lstinline?s?): \lstinline{struct s}\ast }
\end{prooftree}
\qquad
\begin{prooftree}[right label template=\tiny\textsc{(\inserttext)},center=false]
\lstinline?struct?\lstinline?s?
 \hypo{ \Gamma \vdash \lstinline{e}_0:\lstinline?struct?\lstinline?s?\ast}
\infer1[T-mfree]{ \Gamma \vdash \MFREE(\lstinline{e}_0): \lstinline{void}\ast }
\end{prooftree}

\vspace{5pt}
\hrule
   \caption{LC typing.}
  \label{fig:LC-typing}
\end{figure}

\begin{exmp}[Some well-typed LC programs]\label{exa:LC-programs-typing}
The six programs illustrated  in \Cref{fig:queue-var1,fig:queue-var2,fig:queue-var3-4,fig:queue-var5-6} are well-typed LC programs.
\end{exmp}

\begin{fact}[On the typing of LC and ANSI C]\label{rem:LC-typing}
Let \lstinline?FD = T f(?$\overline{\lstinline?T?}\ \overline{\lstinline?x?}$\lstinline?) {?$\cdots$\lstinline?}? be a function definition, let
\prototypeof(\lstinline?FD?)  denote  \lstinline?T f(?$\overline{\lstinline?T?}\ \overline{\lstinline?x?}$\lstinline?);? and, let
\prototypeof($\overline{\lstinline?FD?}$) denote \prototypeof(\lstinline?FD?$_1$)$\cdots$\prototypeof(\lstinline?FD?$_n$).\\
If $\overline{ \lstinline!SD! }  \; \overline{ \lstinline!FD!}$ is a \emph{well-typed} LC program, i.e., 
$\vdash \overline{ \lstinline!SD! }  \; \overline{ \lstinline!FD!} \; \textsc{ok}$ holds,
then
$\overline{ \lstinline!SD! }$ \prototypeof($\overline{\lstinline?FD?}$) $\overline{ \lstinline!FD!}$ 
 is a well-typed C24 program~\cite{C24-standard}.
\end{fact}


\section{Colored Lightweight C}\label{sec:SPL-language}
In this section we  introduce  \emph{Colored LC (CLC)}, a core calculus for SPLs of (L)C programs.
\FDadd{To represent the annotations, we adopt} the approach proposed for Colored Featherweight Java (CFJ)~\cite{kastner2021acm}.

\subsection{CLC Syntax}


The C preprocessor  \FDadd{allows to remove arbitrary} code fragments \FDadd{(even fragments starting in the body of a function and ending in the body of another function),} making the process prone to syntactical errors. 
\FDadd{As in CFJ~\cite{kastner2021acm}, in 
 CLC, we adopt a neat discipline on what can be removed. Namely, we allow}
 programmers to remove some of the \tcbox{highlighted in gray} elements  in Figure~\ref{fig:CLC-syntax}:  
any struct definition, 
any member of a struct, any function definition, any formal parameter of a function definition, any actual parameter of a function application, and any expression in an expression sequence.
It is worth to note that in the concrete syntax the list of expressions include commas that can be annotated, but in the abstract one we neglect their presence:
  they can be faced in accord to Pattern 2 of (cf. Remark~\ref{rem:SPL-patterns}) and they can be managed uniformly as list-elements.

 \begin{figure}
  \begin{tabular}{lcl@{\qquad}l}
  \hline\\[-5pt]
   \lstinline!Prg! & ::=  & \tcbox{$\overline{ \lstinline!SD! } $} \tcbox{$ \overline{ \lstinline!FD!}$}
    &  Program  \\
    \lstinline!SD! & ::=  &\lstinline?struct s?\{\tcbox{$\overline{ \lstinline!T! }\ \overline{\lstinline!m!};$}\};    
    &  Struct definition  \\ 
    \lstinline!FD! & ::=  &\lstinline?T f(?\tcbox{$\overline{\lstinline?T?}\ \overline{\lstinline?x?}$}\lstinline?) { return e; }?     &  Function definition  \\
   \lstinline!T! & ::=  &  \lstinline?int?  $\;\mid$ \lstinline?void?$\ast$ $\mid$ \lstinline?struct?\lstinline?s?$\ast$  & Type \\  
    \lstinline!e! & ::=  &   n $\;\mid \;\NULL\; \mid$ \lstinline?x?                        
                           $\mid$ \lstinline?f? (\tcbox{$\bar{\lstinline!e!}$}) 
                           $\mid$ \lstinline? e->m ? 
                           $\mid$ \lstinline?e->m=e?                              
                           $\mid$ \lstinline!e?e:e!    
                           $\mid$ (\tcbox{$\tilde{\lstinline'e'}$})                             
                                    &    \\
 &  &    
                           $\mid$ \uop \lstinline?e?      
                           $\mid$ \lstinline?e? \bop \lstinline?e?
                           $\mid$  \MALLOC(\lstinline?struct?\lstinline?s?)
                           $\mid$  \MFREE(\lstinline?e?)
    &  Expression  \\
  \hline      
 \end{tabular} 
  \caption{LC syntax (the same as in Figure~\ref{fig:LC-syntax}) where the parts \tcbox{highlighted in gray} are those that can be annotated.}
  \label{fig:CLC-syntax}
\end{figure}

We model code annotations in the CB by using an \emph{annotation table}, that maps each occurrence of an annotable code fragment (identified by its location, e.g.,
in the text or in the syntax tree, of the CB)\footnote{So, distinct occurrences of a same code fragment can have different annotations.} to a propositional formula over features specifying whether 
the occurrence of the code fragment must be kept when generating  variants. Therefore a CLC program, that represents an SPL of LC programs,  is a triple $(\Phi,\lstinline{Prg},\CK)$ where: 
\begin{itemize}
\item
$\Phi$ is an FM  (see \Cref{def:FM}),
\item
\lstinline{Prg} is an LC program, which represents the CB, and
\item
\CK is an annotation table, which expresses CK.
\end{itemize}
To simplify the
formalization, in what follows we always assume a fixed SPL $(\Phi,\lstinline{Prg},\CK)$.

We use the metavariables $a$ and $b$ to refer to arbitrary annotable code fragments.
For every annotable fragment $a$ of the CB and for every product $p$ of the FM, the propositional formula $\CK(a)$ is such that if $p\models \CK (a)$  (where $p$ is used as the interpretation which maps to $1$ the features in $p$ and maps to $0$ the other features) then the code fragment $a$ must be kept when generating the variant for $p$, otherwise it  must be removed. The fragments without annotation are implicitly assumed to be annotated by the constant \lstinline{1} (i.e., to be present in all variants).

\begin{remark}[On the CLC annotation table and the \#if  preprocessor directives]\label{rem:CLC-AT}
  It is straightforward to see that the CK information specified by the annotation table can be implemented by the \lstinline?#if?  preprocessor directives (cf.\ \Cref{sec:running-example}).
\end{remark}

\begin{remark}[Syntactic sugar for the CLC annotation table]\label{rem:CLC-AT-sugar} 
In the examples, we assume that it is also possible to annotate any occurrence of a binary logic operator together with either its left or its right argument; i.e., each occurrence of an expression ``\lstinline?e?$_1$ \bop\ \lstinline?e?$_2$'', where \bop\ is a binary logic operator, may be  either \emph{left-annotated} as  
``\tcbox{\lstinline?e?$_1$ \bop} \lstinline?e?$_2$'' or \emph{right-annotated} as ``\lstinline?e?$_1\!$ \tcbox{\bop\ \lstinline?e?$_2$}''. It is worth observing that this kind of annotations:
\begin{itemize}
\item
can be expressed by using the \lstinline?#if?  preprocessor directives, as done in the code in lines 40-42 and 44-46 in \Cref{fig:queue-spl}; and
\item
can be formalized by introducing some complexity in the definition of the annotation table. 
\end{itemize}
To keep the presentation simpler, we have preferred  not to formalize left- and right-annotations of binary logic operators  and to consider:
\begin{itemize}
\item
 ``$\lstinline?e?_1\ \bop\ \lstinline?e?_2$'' with $\CK(\lstinline?e?_1\ \bop)=\psi$ as syntactic sugar for the expression sequence ``$(\lstinline?e?_1\ \bop\ \lstinline?e?_2, \lstinline?e?_2)$''
 with $\CK(\lstinline?e?_1\ \bop\ \lstinline?e?_2)=\psi$ and $\CK(\lstinline?e?_2)=!\psi$; and
 \item
  ``$\lstinline?e?_1\ \bop\ \lstinline?e?_2$'' with $\CK(\bop\ \lstinline?e?_2)=\psi$ as syntactic sugar for the expression sequence ``$(\lstinline?e?_1\ \bop\ \lstinline?e?_2, \lstinline?e?_1)$''
 with $\CK(\lstinline?e?_1\ \bop\ \lstinline?e?_2)=\psi$ and $\CK(\lstinline?e?_2)=!\psi$.
\end{itemize}
\end{remark}

\begin{exmp}[The Queue SPL]\label{exa:CLC-QueueSPL}
Let $(\Phi_{\mathit{queue}},\lstinline{Prg}_{\mathit{queue}},\CK_{\mathit{queue}})$ be the  CLC SPL such that:
\begin{itemize}
\item
$\Phi_{\mathit{queue}}$ is the FM defined at the beginning of~\Cref{sec:running-example}; 
\item
$\lstinline{Prg}_{\mathit{queue}}$ is the LC program  in \Cref{fig:queue-spl-stripped} -- ignore the colors; and
\item
$\CK_{\mathit{queue}}$ is the annotation table described by the colors in \Cref{fig:queue-spl-stripped}, where:  line 23 and line 25 
contains annotations expressed by using the syntactic sugar described in Remark~\ref{rem:CLC-AT-sugar}; and   line 27 contains the two annotations 
$\CK_{\mathit{queue}}(\lstinline?start$\!\!\!-1$?)=\lstinline?INTERVAL?$ and $\CK_{\mathit{queue}}(\lstinline?end-1?)=\lstinline?INTERVAL?$.
\end{itemize}
The ANSI C encoding of the pair $(\lstinline{Prg}_{\mathit{queue}},\CK_{\mathit{queue}})$ is given  in \Cref{fig:queue-spl}, 
where the two annotations in line 27 of \Cref{fig:queue-spl-stripped} ($\CK_{\mathit{queue}}(\lstinline?start-1?)=\lstinline?INTERVAL?$ and $\CK_{\mathit{queue}}(\lstinline?end-1?)=\lstinline?INTERVAL?$) are expressed by a single \#ifdef preprocessor directive in lines 48-50 of \Cref{fig:queue-spl}.
\end{exmp}

 \begin{notation}[Some abbreviation for propositional formulas involving the annotation table]\label{not:AT-formulas} $\;$
\begin{itemize}
\item
We write ``$\CK(\bar{a})\odot \CK(\bar{b})$'', with $\odot\in \{\Rightarrow,\Leftrightarrow\}$ (see Notation \ref{not:Rightarrow-Leftrightarrow}), as short for 
\begin{center}
``$(\CK(a_1)\odot \CK(b_1)) \,\&\& \,\cdots \,\&\&\, (\CK(a_n)\odot \CK(b_n))$'' 
\end{center}
thus also stating that the sequences $\bar{a}$ and $\bar{b}$ have the same length $n$.
\item
\FDadd{We write} ``$\CK(\lstinline{T})$'' as short for
\begin{itemize}
\item
 ``1'', if $\lstinline{T}\in\{\lstinline{int}, \lstinline{void}\ast\}$, which states that types $\lstinline{int}$ and $\lstinline{void}\ast$
are always present;
and 
\item
``$\CK(\lstinline{Prg(s))}$'', if
 $\lstinline{T}=\lstinline{struct} \lstinline{s}\ast$, which states  that
 the struct  \lstinline{s} is present. 
  \end{itemize}
\item
We write
``$\exists(\tilde{\lstinline!e!})$'' as short for  
``$\CK(\lstinline!e!_1) \,||\, \cdots \,||\, \CK(\lstinline!e!_n)$'',  which states that there is at least one $\lstinline!e!_i$ in $\tilde{\lstinline!e!}$  such that  $\CK(\lstinline!e!_i)$ is true;
i.e., that 
there is at least an  expression of  $\tilde{\lstinline!e!}$ that is not removed in the generation of some variant.
\item 
We write
``$\texttt{neverLast}(k,\bar{\lstinline!e!})$'' as short for ``$!\CK(\lstinline!e!_k)  \,||\, \CK(\lstinline!e!_{k+1}) \,||\, \ldots \,||\, \CK(\lstinline!e!_n)$, which states that 
in every variant $\lstinline!e!_k$ is never the last (not suppressed) expression of $\bar{\lstinline!e!}$, viz., if it is not removed in the generation of a variant then at least one subsequent espression is not removed  in the generation of that variant.
 \end{itemize}
\end{notation}


\subsection{CLC Typing}\label{sec:CLC-typing}

 In order to guarantee that all the variants of a CLC SPL are well-typed LC programs
 we define a family-based type system  for CLC SPLs. 
 
 An \emph{annotated type environment} $\Delta$ is a finite mapping from \FDadd{variables}  to pairs of a type and a proposition formula over features, written 
$\overline{\lstinline?x?}:\overline{\lstinline?T?}\overset{\text{with } \overline{\psi}}{\;}$ which (as usual) is short for
 $\lstinline?x?_1:\lstinline?T?_1 \overset{\text{with }\psi_1}{\;},...,\lstinline?x?_n:\lstinline?T?_n\overset{\text{with }\psi_n}{\;}$  with $n\ge 0$.
The family-based typing judgements are:
\begin{itemize}
\item
$\phi  \deriv \lstinline?Prg?  \; \textsc{ok}$, to be read ``the SPL $(\Phi,\lstinline{Prg},\CK)$  with $\Phi=(\mathcal{F},\phi)$ is well  typed'';
\item
$\theta\deriv \lstinline{SD} \; \textsc{ok}$, to be read ``\lstinline{SD} is well typed under the configuration assumption\footnote{We call configuration assumption the propositional formulas that we use in the family-based typing judgements.} $\theta$'';
\item
$\theta \deriv \lstinline{FD} \; \textsc{ok}$, to be read ``\lstinline{FD} is well typed under the configuration assumption $\theta$''; and
\item
$\theta; \Delta \deriv \lstinline{e}: \lstinline?T?$,  to be read ``expression  \lstinline{e} has type \lstinline?T? under the configuration assumption $\theta$ and the assumptions in the annotated type environment $\Delta$''.
\end{itemize}
 The family-based typing rules are given in \Cref{fig:CLC-typing}, where we use the abbreviations introduced in Notation~\ref{not:AT-formulas} and
in the premises of rule \textsc{FT-prg} we write ``$\phi \,\&\&\, \CK(\overline{\lstinline?SD?}) \deriv \overline{\lstinline?SD?}  \; \textsc{ok}$''
as abbreviation for ``$\phi \,\&\&\, \CK(\lstinline?SD?_1) \vdash \lstinline?FD?_1  \; \textsc{ok}$ $\cdots$ $\phi \,\&\&\, \CK(\lstinline?SD?_n) \vdash \lstinline?FD?_n  \; \textsc{ok}$''
where $n\ge 0$ is length of the sequence  $\overline{\lstinline?SD?}$, 
and similarly for ``$\phi \,\&\&\, \CK(\overline{\lstinline?FD?}) \deriv \overline{\lstinline?FD?}  \; \textsc{ok}$''.
\FDadd{As usual,  $\theta\models\theta'$ means that the propositional formula $\theta\Rightarrow\theta'$ is valid (i.e., it is satisfied by all assignments of truth values to its variables).}

\begin{figure}
  \hrule
  \vspace{2pt}
  \textbf{SPL family-based typing} \hfill \fbox{\tcbox{$\phi$} $\deriv \lstinline?Prg?  \; \textsc{ok}$}
   \vspace{5pt}
  \\
  \begin{prooftree}[right label template=\tiny(\inserttext),separation=0.9em] 
       \hypo{ \lstinline!Prg! \; \textsc{sane}  }
           \hypo{\tcbox{\CK(\lstinline?main?) = \lstinline{1}} }
 \hypo{ \lstinline!Prg! = \overline{\lstinline?SD?}\;\overline{\lstinline?FD?} }
 \hypo{\tcbox{$\phi \,\&\&\, \CK(\overline{\lstinline?SD?})\deriv \overline{\lstinline?SD?}  \; \textsc{ok}$} }
 \hypo{ \tcbox{$\phi \,\&\&\, \CK(\overline{\lstinline?FD?})$} \!\!\deriv \overline{\lstinline?FD?}  \; \textsc{ok}}
\infer5[\textsc{FT-prg}]{ \tcbox{$\phi$}\deriv  \lstinline!Prg!\; \textsc{ok} }
\end{prooftree}
  \vspace{5pt}
 \hrule
      \vspace{2pt}
  \tcbox{\textbf{Structure family-based typing}} \hfill \fbox{\tcbox{$\theta\deriv \lstinline{SD} \; \textsc{ok}$}}
   \vspace{5pt}
   \\
  \tcbox{\begin{prooftree}[right label template=\tiny(\inserttext)]
    \hypo{\theta  \models   \CK(\overline{ \lstinline!T! }\ \overline{\lstinline!m!})  \Rightarrow  \CK(\overline{ \lstinline!T! }) } 
    \infer1[\textsc{FT-struct}]{ \theta \deriv \lstinline?struct s?\{\overline{ \lstinline!T! }\ \overline{\lstinline!m!} \}    \; \textsc{ok}}
  \end{prooftree}
}
  \vspace{5pt}
  \hrule
      \vspace{2pt}
  \textbf{Function  family-based typing} \hfill \fbox{\tcbox{$\theta$} $\deriv \lstinline{FD} \; \textsc{ok}$}
   \vspace{5pt}
  \\
  \begin{prooftree}[right label template=\tiny(\inserttext)]
    \hypo{ \tcbox{$\theta \models \CK(\lstinline!T!_0)$} }
    \hypo{ \tcbox{$\theta \models \CK(\overline{ \lstinline!T! }\ \overline{\lstinline!x!})  \Rightarrow  \CK(\overline{ \lstinline!T! })$} }
 \hypo{  \tcbox{$\theta;$} \overline{\lstinline?x?}  : \overline{\lstinline?T?} \tcbox{$\overset{\text{ with } \CK(\overline{ \lstinline!T! }\ \overline{\lstinline!x!})}{\;}$}   \deriv \lstinline{e}: \lstinline{T}_1 } 
 \hypo{  \lstinline?T?_1  \le\lstinline?T?_0  } 
\infer4[\textsc{FT-fun}]{ \tcbox{$\theta$}  \deriv \lstinline?T?_0 \; \lstinline?f(?\overline{\lstinline?T?} \; \overline{\lstinline?x?}\lstinline?) \{return e?\lstinline?;\}? \; \textsc{ok} }
\end{prooftree}
\vspace{5pt}
\hrule
  \vspace{2pt}
  \textbf{Expression  family-based typing} \hfill \fbox{\tcbox{$\theta;$} $\Delta \deriv \lstinline{e}: \lstinline?T?$}
   \vspace{5pt}
  \\
\begin{prooftree}[right label template=\tiny\textsc{(\inserttext)}]
\infer0[FT-int]{ \tcbox{$\theta;$}\Delta \deriv \text{n}: \lstinline{int} }
\end{prooftree}
\quad
\begin{prooftree}[right label template=\tiny\textsc{(\inserttext)}]
\infer0[FT-null]{ \tcbox{$\theta;$}\Delta \deriv \NULL: \lstinline{void}\ast }
\end{prooftree}
\quad
\begin{prooftree}[right label template=\tiny\textsc{(\inserttext)}]
  \hypo{ \lstinline?x?:\lstinline?T?\tcbox{$\overset{\text{ with } \psi}{\;}$}  \in\Delta}
  \hypo{ \tcbox{$\theta \models \psi$}}
\infer2[FT-par]{  \tcbox{$\theta;$}\Delta \deriv  \lstinline?x? :\lstinline?T? }
\end{prooftree}
\bigskip

\begin{prooftree}[right label template=\tiny\textsc{(\inserttext)},center=false]
\hypo{ \lstinline!Prg!(\lstinline?f?)=\lstinline?T?_0 \; \lstinline?f(?\overline{\lstinline?T?} \; \overline{\lstinline?x?})\{\ldots \}  }
  \infer[no rule]1{ \tcbox{$\theta \models\CK(\lstinline!Prg!(\lstinline?f?))$}     }
 \hypo{ \tcbox{$\theta;$}\Delta \deriv \overline{\lstinline{e}} :  \overline{\lstinline?T?}' \qquad \overline{\lstinline?T?}'\le\overline{\lstinline?T?}}
 \infer[no rule]1{ \tcbox{$\theta \models\CK( \overline{\lstinline{e}}) \Leftrightarrow \CK(\overline{\lstinline?T?} \; \overline{\lstinline?x?})$} }
 \infer2[FT-app]{  \tcbox{$\theta;$} \Delta \deriv \lstinline?f? (\overline{ \lstinline!e!}) : \lstinline!T!_0  }
\end{prooftree}
\quad
\begin{prooftree}[right label template=\tiny\textsc{(\inserttext)},center=false,separation=0.8em]
\hypo{  \tcbox{$\theta;$}\Delta \deriv \lstinline{e}_0: \lstinline{struct s}\ast }
 \hypo{ \lstinline!Prg!(\lstinline{s})(\lstinline{m}) =  \lstinline?T m?}
  \infer[no rule]1{  \tcbox{$\theta \models \CK(\lstinline?T m?)$}  }
\infer2[FT-member]{ \tcbox{$\theta;$} \Delta \deriv \lstinline{e}_0\lstinline{->m} : \lstinline?T? }
\end{prooftree}
\medskip

\begin{prooftree}[right label template=\tiny\textsc{(\inserttext)}]
  \hypo{ \tcbox{$\theta;$}\Delta \deriv \lstinline{e}_0: \lstinline{struct s}\ast }
  \hypo{ \lstinline!Prg!(\lstinline{s})(\lstinline{m}) =  \lstinline?T m?}
   \hypo{ \tcbox{$\theta;$}\Delta \deriv \lstinline{e}_1:\lstinline?T?_1}
    \hypo{  \lstinline?T?_1\le\lstinline?T?}
    \hypo{   \tcbox{$\theta \models \CK(\lstinline?T m?)$}  }
\infer5[FT-assign]{ \tcbox{$\theta;$}\Delta \deriv \lstinline{e}_0\lstinline{->m}=\lstinline{e}_1 : \lstinline?T? }
\end{prooftree}
  \bigskip

\begin{prooftree}[right label template=\tiny\textsc{(\inserttext)},center=false]
 \hypo{  \tcbox{$\theta;$}\Delta \deriv \lstinline{e}_0: \lstinline?T?_0 }
  \hypo{\lstinline?T?_3=\maxtype\{\lstinline?T?_1,\lstinline?T?_2\} }
  \hypo{ \tcbox{$\theta;$}\Delta \deriv \lstinline{e}_1:\lstinline?T?_1  }
  \hypo{\tcbox{$\theta;$}\Delta \deriv \lstinline{e}_2:\lstinline?T?_2}
\infer4[FT-cond]{ \tcbox{$\theta;$}\Delta \deriv \lstinline{e}_0\lstinline{?e}_1\lstinline{:e}_2 : \lstinline!T!_3 }
\end{prooftree}
\bigskip

  \begin{prooftree}[right label template=\tiny\textsc{(\inserttext)},center=false]
    \hypo{  
      \tilde{\lstinline!e!}=\lstinline!e!_1,\ldots,\lstinline!e!_n
    }
    \infer[no rule]1{
      \tcbox{$\theta\models  \exists(\lstinline!e!_1,\ldots,\lstinline!e!_n)$}
   }
    \hypo{  
      \tcbox{$\theta \&\&\CK(\lstinline!e!_1);$}\Delta \deriv \lstinline!e!_1 :\lstinline!T!_1  \quad \cdots \quad  \tcbox{$\theta \&\&\CK(\lstinline!e!_n);$}\Delta \deriv \lstinline!e!_n :\lstinline!T!_n 
       \qquad  
       \lstinline!T!= \lstinline!T!_n
    }
    \infer[no rule]1{
      \tcbox{$\text{ $\forall i\in\{1....,n\}$, if } (\lstinline!T!_i \neq \lstinline?T?) \text{ then }
      (\theta\models \texttt{neverLast}(i,\tilde{\lstinline!e!}))$}
   }
\infer2[FT-seq]{ \tcbox{$\theta;$}  \Delta \deriv (\tilde{\lstinline!e!}) :  \lstinline?T? }
\end{prooftree}
\bigskip

\begin{prooftree}[right label template=\tiny\textsc{(\inserttext)},center=false]
  \hypo{ \tcbox{$\theta;$}\Delta \deriv \lstinline{e}_0 :\lstinline?T?_0}
 \infer[no rule]1{  \typeof(\texttt{uop})=(\lstinline?T?_0)\!\!\rightarrow\!\!\lstinline?int?}
\infer1[FT-uop]{   \tcbox{$\theta;$}\Delta \deriv \texttt{uop}(\lstinline!e!_0) :  \lstinline?int? }
\end{prooftree}
\qquad
\begin{prooftree}[right label template=\tiny\textsc{(\inserttext)},center=false]
   \hypo{ \tcbox{$\theta;$}\Delta \deriv \lstinline{e}_1 : \lstinline?T?_1}
   \infer[no rule]1{  \lstinline?T?_3=\maxtype\{\lstinline?T?_1,\lstinline?T?_2\} } 
    \hypo{ \tcbox{$\theta;$}\Delta \deriv \lstinline{e}_2 : \lstinline?T?_2  }
\infer[no rule]1{\typeof(\texttt{bop})  =(\lstinline?T?_3,\lstinline?T?_3)\!\!\rightarrow\!\! \lstinline?int?  }
\infer2[FT-bop]{\tcbox{$\theta;$}\Delta \deriv \lstinline{e}_1 \;\texttt{bop}\; \lstinline{e}_2 : \lstinline?int? }
\end{prooftree}
\bigskip

\begin{prooftree}[right label template=\tiny\textsc{(\inserttext)},center=false]
  \hypo{\lstinline?s?\in\domof(\overline{ \lstinline!SD!})}
  \hypo{  \tcbox{$\theta \models \CK(\lstinline?Prg(s)?)$} }
  \infer2[FT-malloc]{ \tcbox{$\theta;$}\Delta \deriv \MALLOC(\lstinline?struct?\lstinline?s?): \lstinline{struct s}\ast }
\end{prooftree}
\qquad
\begin{prooftree}[right label template=\tiny\textsc{(\inserttext)},center=false]
\hypo{ \tcbox{$\theta;$}\Delta \deriv \lstinline{e}_0:\lstinline?struct?\lstinline?s?\ast}
\infer1[FT-mfree]{ \tcbox{$\theta;$}\Delta \deriv \MFREE(\lstinline{e}_0): \lstinline{void}\ast }
\end{prooftree}
   \vspace{5pt}
  \hrule
   \caption{Family-based typing of a CLC  SPL  $(\Phi,\lstinline{Prg},\CK)$  with $\Phi=(\mathcal{F},\phi)$; the extensions w.r.t.\  
    LC typing (in \Cref{fig:LC-typing}) are \tcbox{highlighted in grey.} }
  \label{fig:CLC-typing}
\end{figure}

 The CLC family-based  type system (which types a CLC SPL as whole) checks some \emph{presence-in-variant} conditions in addition to the standard checks done by the LC type system (which types a single LC program).
 For instance, we have to express condition like  ``whenever code fragment $a$ is present (in a variant), then also code fragment $b$ is present (in the same variant)'',
 or in other words ``whenever the annotation of $a$ evaluates to 1 (modulo a product $p$) then also the annotation of $b$ evaluates to 1 (modulo the same product $p$)''.

Namely, each family-based typing rule is obtained from an LC typing rule (cf.\ \Cref{fig:LC-typing}) by adding checks (\tcbox{highlighted in grey}) involving the  FM and CK (i.e., the annotation table) ---
in particular each family-based typing rule \textsc{FT-<name>} is obtained from the typing rule \textsc{T-<name>}.

The  extra-conditions checked by CLC SPL family-based typing follow.
 \begin{description}
 \item[P1] (in rule \textsc{FT-prg}) the function \lstinline{main} must be present.
  \item[S1] (in rule \textsc{FT-struct}) 
 If a member \lstinline{T m} of \lstinline{struct} \lstinline{s} is present  then its type \lstinline{T} is present.
 \item[F1] (in rule \textsc{FT-fun})  If  function definition \lstinline{T$_0$ f} is present then its return type  \lstinline{T$_0$} is present.
 \item[F2] (in rule \textsc{FT-fun}) If an actual parameter \lstinline{T x} (of a present function)  is  present then its type \lstinline{T} is present.
 \item[E1] (in rule \textsc{FT-par}) If the occurrence of a function formal parameters \lstinline?x? is  present in the body of a (present) dunction then
 declaration of  \lstinline{x} in the header of the function is present.
 \item[E2] (in rule \textsc{FT-app}) If a function call expression \lstinline?f?($\bar{\lstinline!e!}$)  is present then: (i) the function \lstinline?f? is present
   and (ii) the types and the number of expressions in the sequence of actual paremeter list $\bar{\lstinline!e!}$ match those of the formal parameter list in the definition of \lstinline?f?.
 \item[E3] (in rules \textsc{FT-member} and \textsc{FT-assign})
If an expression \lstinline?e->m? is present then the  structure \lstinline{s} such that
   $\lstinline{struct s}\ast$  is the type of \lstinline{e} is present and containts a member \lstinline{m}.
   \item[E4] (in rule \textsc{FT-seq}) If an expression sequence $(\lstinline{e}_1,...,\lstinline{e}_n)$ is present    then at least  an expression of the sequence is present. 
 \item[E5] (in rule \textsc{FT-seq}) If an expression $\lstinline{e}_i$ $(1\le i\le n)$ in an sequence $(\lstinline{e}_1,...,\lstinline{e}_n)$ has a type $\lstinline{T}_i$ that is different
 from the type  $\lstinline{T}_n$ of the last expression of the sequence, then for  some $j\in\{i+1,...,n\}$ the expression  $\lstinline{e}_j$  is present (so $\lstinline{e}_i$  cannot be the last included in the list in any variant).
\item[E6] (in rule \textsc{FT-malloc}) If an expression \MALLOC(\lstinline?struct?\lstinline?s?) is  present then the struct \lstinline{s} is present.
 \end{description}


  Let $(\Phi,\lstinline{Prg},\CK)$ be an SPL. If $p$ is a valid product then $\lsem \lstinline?Prg? \rsem_p$ 
    is a program in accord
    with the grammar of Figure~\ref{fig:LC-syntax} (possibly not well-typed).
 
\begin{exmp}[Family-based typing of the Queue SPL]\label{exa:CLC-QueueSPL-typing}
The Queue SPL given in Example~\ref{exa:CLC-QueueSPL} is a well-typed CLC SPL.
\end{exmp}

\subsection{Variant Generation}

In order to generate program variants, we have to define the variant generator VG for CLC SPLs.
Given product $p$,  the generator VG descends the syntactic tree of the \lstinline?Prg? (viz. the CB of the SPL)
and removes all code fragments annoted by the annotation table $\CK$ with a formula $\phi$ which evaluates to $0$ under the interpretation identified by $p$.

  Given a code base $\CB$ and a product $p$, we denote $\lsem \CB \rsem_p$ the corresponding variant.
The definition of the variant generator is given in Fig.~\ref{fig:VG}, it is defined by means of two mutual recursive defined functions:
  $\lsem \_ \rsem_p$ suppress the elements of list annoted with formulas evaluated in false by $p$,
while  $\llangle \_ \rrangle_p$  descends the syntax tree and apply $\lsem \_ \rsem_p$ to annotated sequences.

\begin{figure} 
  $$\begin{array}{rcl<{\hspace{3.8mm}\scalebox{.8}{(\textsc{vgList})}}}
   \hline
   \\
    \lsem a_1 \ldots a_n \rsem_p   & := & 
    \left\{
            \begin{array}{ll}
         \bullet & \text{if }  n==0\\                                     
           \llangle  a_1  \rrangle \lsem a_2 \ldots a_n \rsem_p & \text{if } p\models\CK(a_1), n\geq 1\\
            \lsem a_2 \ldots a_n \rsem_p & \text{if } p\not\models\CK(a_1), n\geq 1\\
           \end{array}
           \right.\\
           \\
      \hline
\end{array}$$
  $$\begin{array}[t]{rcl<{\hspace{\fill}\scalebox{.8}{(\textsc{vg\therowVG})\stepcounter{rowVG}}}}
  \\
   \llangle \lstinline?struct$\,$s?\{\overline{ \lstinline!T! }\,\overline{\lstinline!m!};\};\rrangle_p   &:= &  \lstinline?struct$\,$s?\{ \ \lsem \overline{ \lstinline!T! }\ 
   \overline{\lstinline!m!}; \rsem_p  \};
  \\
      \llangle \lstinline?T f(?\overline{\lstinline?T?}\,\overline{\lstinline?x?}\lstinline?){return e;}? \rrangle_p  &:= & \lstinline?T f(?\lsem \overline{\lstinline?T?}\ \overline{\lstinline?x?}\rsem_p\lstinline?){return? \ \llangle \lstinline?e?\rrangle_p ;\}
    \\
    \hline
      \llangle \lstinline?T m?\rrangle_p&:= & \lstinline?T m? 
      \\
    \llangle \lstinline?T x?\rrangle_p&:= & \lstinline?T x?  
        \\
    \hline
  \llangle \lstinline?n?\rrangle_p&:= & \lstinline?n? \\
  \llangle \NULL\rrangle_p&:= &  \NULL \\
  \llangle \lstinline?x?\rrangle_p&:= &  \lstinline?x?\\
  \llangle \lstinline?f? (\bar{\lstinline!e!}) \rrangle_p&:= & \lstinline?f? (\ \lsem \bar{\lstinline!e!}\rsem_p  \ );  \\
  \llangle \lstinline?e->m ? \rrangle_p&:= &  \llangle \lstinline?e?\rrangle_p\lstinline?->m? \\
  \llangle \lstinline?e->m=e?   \rrangle_p&:= & \llangle\lstinline?e?\rrangle_p\lstinline?->m? = \llangle \lstinline?e?\rrangle_p\\
  \llangle  \lstinline!e?e:e!   \rrangle_p&:= & \llangle\lstinline?e?\rrangle_p \ \lstinline'?'\ \llangle\lstinline?e?\rrangle_p  \ \lstinline':'\  \llangle\lstinline?e?\rrangle_p \\
  \llangle (\tilde{\lstinline'e'})  \rrangle_p&:= &  (\lsem \bar{\lstinline!e!}\rsem_p  )  \\
  \llangle \uop \lstinline? e? \rrangle_p&:= & \uop \llangle\lstinline?e?\rrangle_p\\
  \llangle \lstinline?e ? \bop \lstinline? e?      \rrangle_p&:= & \llangle\lstinline?e?\rrangle_p \;\bop\;\llangle\lstinline?e?\rrangle_p \\
  \llangle  \MALLOC(\lstinline?struct?\lstinline?s?) \rrangle_p&:= & \MALLOC(\lstinline?struct?\lstinline?s?) \\
      \llangle \MFREE(\lstinline?e?)\rrangle_p&:= & \MFREE(\llangle\lstinline?e?\rrangle_p) \\
      \\
 \hline
\end{array}$$
\caption{Variant generator of the CLC  SPL $(\Phi,\overline{\lstinline!SD! }\;\overline{\lstinline!FD!},\CK)$
 applied to the product $p$.}
\label{fig:VG}
 \end{figure}

\begin{exmp}[The variants of the Queue SPL]\label{exa:CLC-QueueSPL-variants}
Consider the well-typed Queue SPL given in Example~\ref{exa:CLC-QueueSPL}. Its variants, generated by applying the VG  to the annotated code base  illustrated in \Cref{fig:queue-spl-stripped}, are  the six well-typed LC programs given in Example~\ref{exa:LC-programs-typing}.
\end{exmp}


\section{Properties}\label{sec:SPL-type-safety}
An SPL of programs written in a given programing language PL is \emph{type safe} if all its variants are well-typed programs
according to the type system of PL.
In this section we prove \emph{type safety} for CLC, i.e., 
well-typed CLC SPLs  are type safe.
\FDaddCR{We first prove (in \Cref{sec:global}) that the code base of a well-typed CLC SPL is a well-typed LC program and some auxiliary lemmas; then
(in \Cref{sec:CLC-type-safety}) we prove some basic guaranntess on variant generation and the type-safety result.}

\subsection{The Code Base of a Well-typed CLC SPL is a Well-typed LC Program \FDaddCR{and  Some Auxiliary Lemmas}}\label{sec:global}

\begin{thm}[The code base of a well-typed CLC SPL is a well-typed LC program]\label{lem:global}
Consider a CLC SPL $(\Phi,\lstinline{Prg},\CK)$  with $\Phi=(\mathcal{F},\phi)$.
  If $\phi\deriv \lstinline?Prg?  \; \textsc{ok}$ then $\vdash \lstinline?Prg?  \; \textsc{ok}$. 
\end{thm}
\begin{proof}
    First recall that the LC typing rules, in  \Cref{fig:LC-typing}, can be otained from the family-based typing rules, in \Cref{fig:CLC-typing}, by ignoring
    the parts \tcbox{highlighted in grey.} Namely, by:
    (i) neglecting the premises about the typing of structures in \textsc{FT-prg} (indeed, LC just requires the sanity condition about structures); and 
    (ii)  neglecting the premises that involve $\models$ in the other rules (i.e. checks about annotations). 
  Then, the proof is straightforward by induction on the derivation $\mathcal{D}$ proving $\phi\deriv \lstinline?Prg?  \; \textsc{ok}$.
\end{proof}

\FDaddCR{We now present some auxiliary lemmas.}

 Let $\mathcal{D}$  be the derivation proving  $\phi\deriv \lstinline?Prg?  \; \textsc{ok}$.
   It is worth to remark that the derivation $\mathcal{D}'$ of  $\vdash \lstinline?Prg?  \; \textsc{ok}$ built in the proof of
   \Cref{lem:global}
 is, roughly speaking, obtained from $\mathcal{D}$ by  removing all gray-identified (in accord to \Cref{fig:CLC-typing}) parts of its rules and then replacing each  rule \textsc{FT-<name>} with rule \textsc{T-<name>}.

\begin{lemma}[On subtyping an annotations in family-based typing]\label{lem:annotSubtyping}
    Consider a CLC SPL $(\Phi,\lstinline{Prg},\CK)$  and two types
    $\lstinline?T?_0, \lstinline?T?_1$ ocurring in \lstinline{Prg}.
    If $\theta\models\CK(\lstinline?T?_1)$ and $\lstinline?T?_0 \le\lstinline?T?_1$ then $\theta\models\CK(\lstinline?T?_0)$.
\end{lemma}
\begin{proof}
  First observe that if $\lstinline?T?_0 \le\lstinline?T?_1$ then the possible cases are as follows:
  (i) if $\lstinline?T?_0 =\lstinline?int?$ then  $\lstinline?T?_1=\lstinline?int?$;
  (ii) if $\lstinline?T?_0=\lstinline?struct?\lstinline?s?\ast$ for some $\lstinline?struct?\lstinline?s?$, then  $\lstinline?T?_1=\lstinline?structs?\lstinline?s?\ast$;
  and,  (iii) if $\lstinline?T?_0 =\lstinline?void?\ast$ then, either $\lstinline?T?_1=\lstinline?void?\ast$ or $\lstinline?T?_1=\lstinline?struct?\lstinline?s?\ast$ for some struct. In all cases the proof is straightforward, in accord to Notation~\ref{not:AT-formulas}.
\end{proof}

 Let $J$ and $J'$ be two judgment occurrences of the derivation $\mathcal{D}$ of  $\phi\deriv \lstinline?Prg?  \; \textsc{ok}$.
 We say that a statement  $J'$  is a descendant of $J$ ($J$ is an ascendant of $J'$) whenever $J'$ occurs in the (sub)derivation of $J$.

\begin{lemma}[On the formulas on the left of  family-based typing judgment]\label{lem:annotDerivation}
Consider a CLC SPL $(\Phi,\lstinline{Prg},\CK)$  with $\Phi=(\mathcal{F},\phi)$.
Let $J$ and $J'$ be two judgment 
occurrences  in the derivation of $\phi\deriv \lstinline?Prg?  \; \textsc{ok}$ such that $J'$  is a descendant of $J$.
If $\theta_J$ is the leftmost formula of a judgment $J$ and $\theta'_J$ is the leftmost formula  of  $J'$ 
then  $\theta'_J\models\theta_J$.
\end{lemma}
\begin{proof}
The proof is done by induction on the derivation of  $\phi\deriv \lstinline?Prg?  \; \textsc{ok}$. Straighforwardly, 
 each leftmost formula $\theta'$ in premises is: (i) the same leftmost formula $\theta$
 of the judgment in the conclusion of the rule; or, (ii) the conjunction of $\theta$ with another formula.
Therefore, if  $\theta_J$ is the leftmost formula of a judgment $J$ having $J'$ as descendant (leftmost annotated $\theta'_J$), then  $\theta'_J\models\theta_J$.
\end{proof}

\begin{lemma}[On family-based typing of expressions]\label{lem:exprType}
     Consider a CLC SPL $(\Phi,\lstinline{Prg},\CK)$  with $\Phi=(\mathcal{F},\phi)$.
If  $\theta; \Delta \deriv \lstinline?e?  : \lstinline?T?$ is included in
  a derivation $\mathcal{D}^{\lstinline?Prg?}$ proving $\phi\deriv \lstinline?Prg?  \; \textsc{ok}$
  then  $\theta\models\CK(\lstinline?T?)$.
\end{lemma}
\begin{proof}
  If $\mathcal{D}$ be a sub-derivation proving  $\theta; \Delta \deriv \lstinline?e?  : \lstinline?T?$ then, the proof follows by induction on the last rule used in $\mathcal{D}$. 
  Note that $\mathcal{D}$ is a sub-derivation included in $\mathcal{D}^{\lstinline?Prg?}$ by a unique rule \textsc{FT-fun} used immediately over the unique application of  \textsc{FT-prg},  viz. the judgment $\theta; \Delta \deriv \lstinline?e?  : \lstinline?T?$  is about a fragment of the body of a function  $\lstinline?f?_\mathcal{D}$.
  We call this unique  parent rule application the \textsc{FT-fun}-parent of $\mathcal{D}$.
    \begin{itemize}
    \item Cases \textsc{FT-int}, \textsc{FT-null}, \textsc{FT-bop}, \textsc{FT-uop} and \textsc{FT-mfree} are trivial, since we assumed $\CK(\lstinline{int})=\CK(\lstinline{void}\ast)=\lstinline{1}$.
    \item  The case \textsc{FT-par} follows because it verifies $\theta \models \psi$, or more precisely $\theta \models\CK(\lstinline?T?\lstinline?x?)$ for some \lstinline?T?\lstinline?x?.
     Indeed, $\psi$ is the with-labeling formula of a $\lstinline?x?  : \lstinline?T? \overset{\text{ with } \CK(\lstinline!T! \lstinline!x!)}{\;}$ added to the type environment by the  \textsc{FT-fun}-parent of $\mathcal{D}$.  
      This parent rule contains, in its premises,
      $\theta' \models \CK(\lstinline!T! \lstinline!x!)  \Rightarrow  \CK( \lstinline!T!)$ where $\theta'=\phi \&\& \CK(\lstinline!Prg!(\lstinline?f?))$,
      because \textsc{FT-fun} is immediately over the rule \textsc{FT-prg}.
      Since $\theta$ is annotating a descendant judgment, then by Lemma~\ref{lem:annotDerivation} it follows that $\theta\models\theta'$. Thus  $\theta \models \CK(\lstinline!T! \lstinline!x!)  \Rightarrow  \CK( \lstinline!T!)$,
  and we  conclude $\theta \models\CK( \lstinline!T!)$.  
    \item  In case \textsc{FT-app}, the expression $\lstinline?e?$ has to have the shape $\lstinline?f? (\overline{ \lstinline!e!})$ with  premises
      $\theta \models\CK(\lstinline!Prg!(\lstinline?f?))$
      that    ensures the presence of $\lstinline!Prg!(\lstinline?f?)=\lstinline?T?_0 \; \lstinline?f(?\overline{\lstinline?T?} \; \overline{\lstinline?x?})\{\ldots \}$ whenever $\theta$ holds (where $T_0$ is the type of $\lstinline?f? (\overline{ \lstinline!e!})$).
Note that the \textsc{FT-fun}-parent of  $\mathcal{D}$ is used to type a function $\lstinline?f?_\mathcal{D}$ which is possibly different from~$\lstinline?f?$.
Nevertheless, the derivation $\mathcal{D}^{\lstinline?Prg?}$ contains a  subderivation with root a rule \textsc{FT-fun} for each function defined in $\lstinline?Prg?$.
In particular, there is a subderivation  for $\lstinline?f?$  and it includes between its premises: $\theta' \models \CK(\lstinline!T!_0)$
where $\theta'=\phi \&\& \CK(\lstinline!Prg!(\lstinline?f?))$, because \textsc{FT-fun} is immediately over the rule \textsc{FT-prg}.
Clearly $\theta\models\theta'$ and, so  $\theta\models \CK(\lstinline?T?_0)$ follows.
    \item The cases \textsc{FT-member}, \textsc{FT-assign} are similar. They have premise  $\theta \models \CK(\lstinline?T m?)$  about a member of  $\lstinline?struct?\lstinline?s?$.
    Nevertheless, the sanity hypothesis included in $\mathcal{D}^{\lstinline?Prg?}$ by \textsc{FT-Prog} ensures that, there is necessarily a subderivation starting with the rule \textsc{FT-struct}, that verifies the typing of $\lstinline?struct?\lstinline?s?$. This latter has premise $\theta' \models \CK(\lstinline!T!\ \lstinline!m!)  \Rightarrow  \CK(\lstinline!T!)$ where $\theta'=\phi \&\& \CK(\lstinline!Prg!(\lstinline?struct?\lstinline?s?))$. Thus, $\theta\models\theta'$ and the proof easily follows.
   \item Remaining cases follow by induction.
    \end{itemize}
\end{proof}

\subsection{Basic Guarantees on Variant Generation \FDaddCR{and Type Safety for CLC}}\label{sec:CLC-type-safety}

Next theorems ensures that the SPL typing  ensures the coherence of all valid variants.

Until now, we used $\lstinline{Prg}$ as a mapping so that, $\domof(\lstinline{Prg})$ denote its domain. We extend that notation to include ground types, namaly we use  $\extDom(\lstinline{Prg})$ to denote $\domof(\lstinline{Prg})\cup \{ \lstinline!int!, \lstinline!void!\ast\}$.

\begin{thm}[Safety for annotated struct]\label{lem:annotStruct}
  Let $(\Phi,\lstinline{Prg},\CK)$ be an SPL  with $\Phi=(\mathcal{F},\phi)$ such that, both $p$ is a valid product and
  $\lstinline?struct?\lstinline?s?\{ \lstinline!T!_1 \lstinline!m!_1 \ldots \lstinline!T!_n \lstinline!m!_n\}$ is in $\lstinline{Prg}$.
It holds that:
   (i)~$p\models \CK(\lstinline?struct?\lstinline?s?)$ iff $\lstinline?struct?\lstinline?s? \in \domof(\lsem\lstinline{Prg}\rsem_p)$;
    and,
   (ii)~if  
   $p\models \CK(\lstinline?struct?\lstinline?s?)$ then:
   $p \models \CK(\lstinline!T!_i \lstinline!m!_i)$ iff  $\lstinline!T!_i \lstinline!m!_i$ is included in
   $\llangle \lstinline?struct?\lstinline?s?\{\overline{ \lstinline!T! }\ \overline{\lstinline!m!} \} \rrangle_p$;
(iii)~if  $\phi\deriv \lstinline?Prg?  \; \textsc{ok}\ $, $p\models \CK(\lstinline?struct?\lstinline?s?)$ and $p \models \CK(\lstinline!T!_i \lstinline!m!_i)$ then
 $\lstinline!T!_i\in \extDom(\lsem\lstinline{Prg}\rsem_p)$, 
  for all $i$.
\end{thm}
\begin{proof}
  Point (i) follows because, in accord to the rule \textsc{vgList} of the variant generation, $\lsem \_ \rsem_p$ is applied to the list of struct-definitions, and we know that $p\models \CK(\lstinline?struct?\lstinline?s?)$. Likewise, $\lstinline!T!_i \lstinline!m!_i$ is not removed from $\llangle \lstinline?struct?\lstinline?s?\{\overline{ \lstinline!T! }\ \overline{\lstinline!m!} \} \rrangle_p$ because the rule \textsc{vg2}, $\lsem \_ \rsem_p$ is driven by $p \models \CK(\lstinline!T!_i \lstinline!m!_i)$: so  (ii) is proved.
 Point (iii) follows because, in the derivation of  $\phi\deriv \lstinline?Prg?  \; \textsc{ok}\ $ is necessarily concluded with the rule  \textsc{FT-prg}.
  In its premises, there is  $\phi \,\&\&\, \CK(\overline{\lstinline?SD?}) \deriv \overline{\lstinline?SD?}  \; \textsc{ok}\ $ that includes  $\phi \,\&\&\, \CK(\!\lstinline?struct?\lstinline?s?\!) \deriv  \lstinline?struct?\lstinline?s?\{\overline{ \lstinline!T! }\ \overline{\lstinline!m!} \} \; \textsc{ok}\ $ as instance.
  This premise,  in turn, is proved by a rule \textsc{FT-struct} having premises $ \phi \,\&\&\, \CK(\lstinline?struct?\lstinline?s?) \models \CK(\overline{ \lstinline!T! }\ \overline{\lstinline!m!}) \Rightarrow \CK(\overline{ \lstinline!T! })$. Therefore, the proof follows from hypotheses on $p$.
\end{proof}

We write $\lstinline!e!_0\Subset \lstinline?e?_1$ to mean that the location $\lstinline!e!_0$ is included in $\lstinline?e?_1$. Let $\lstinline?e?_{\lstinline?f?}$ be the body of a function~$\lstinline!f!$.
The presence/absence of the location $\lstinline!e!'$ in  the variant generated form a valid product $p$
depends the annotations of many syntactical elements: (i) the annotation of $\lstinline!Prg(f)!)$ (ii) the annotation of $\CK( \lstinline!e!')$; and, also, (iii)~from
the annotations of expression-lists in which $\lstinline!e!'$ is nested in.
More precisely, we denote  $\mathtt{nested}_{\lstinline?f?}(\lstinline!e!')$ the following set of annotable locations 
$\{ (\tilde{\lstinline?e?})\Subset \lstinline?e?_{\lstinline?f?}   \mid  \lstinline!e!' \Subset (\tilde{\lstinline?e?})\} \cup
\{ \lstinline?g?(\bar{\lstinline?e?})\Subset \lstinline?e?_{\lstinline?f?}   \mid  \lstinline!e!' \Subset \lstinline?g?(\bar{\lstinline?e?})\}$.

\begin{lemma}[Function Generation]\label{lem:funGen}
  Let $(\Phi,\lstinline{Prg},\CK)$ be an SPL such that, both $p$ is a valid product and
  $\lstinline?T?\lstinline?f(T?_1\lstinline?x?_1\ldots\lstinline?T?_n\lstinline?x?_n\lstinline?){return e;}?$ is in $\lstinline{Prg}$.
  It holds that:
   (i)~$p\models \CK(\lstinline?f?)$ iff $\lstinline?f? \in \domof(\lsem\lstinline{Prg}\rsem_p)$;
   (ii)~if  $p\models \CK(\lstinline?f?)$ then:
   $p \models \CK(\lstinline?T?_i\lstinline?x?_i)$ iff  $\lstinline?T?_i\lstinline?x?_i$ is included (not removed) in
   $\llangle \lstinline?T?\lstinline?f(T?_1\lstinline?x?_1\ldots\lstinline?T?_n\lstinline?x?_n\lstinline?){return e;}? \rrangle_p$; and,
 (iii)~if
   $p\models \CK(\lstinline?f?)$ and  $\mathtt{nested}_{\lstinline?f?}(\lstinline!e!')=\{\lstinline!e!_1,\ldots,\lstinline!e!_n\}$ then: 
$p\models (\CK(\lstinline?e?_1)\&\& \ldots \&\&\CK(\lstinline?e?_n))$
   iff
   $\lstinline!e!'$ is included in $\llangle \lstinline?T?\lstinline?f(T?_1\lstinline?x?_1\ldots\lstinline?T?_n\lstinline?x?_n\lstinline?){return e;}? \rrangle_p$.
\end{lemma}
\begin{proof}
  Point (i) and (ii) are similar to that of the \Cref{lem:annotStruct}.
  The point (iii) follows by standard (mutual) induction on the (mutual recursive) definition of variant generator.
\end{proof}

Let $\Delta$ be $\overline{\lstinline?x?}  : \overline{\lstinline?T?} \overset{\text{ with } \CK(\overline{ \lstinline!T! }\ \overline{\lstinline!x!})}{\ }$ then,
we denote $\lsem\Delta\rsem_p=\{ \lstinline?x:T? \mid   \lstinline?x:T?\overset{\text{ with } \CK(\lstinline!T!\lstinline!x!)}{\ }\in\Delta \text{ , } p\models \CK(\lstinline!T!\lstinline!x!)
\}$, namely the pairs \lstinline?x:T? that are used in the product $p$ are available as parameter of the considered function.

\begin{thm}[Safety for annotated expression]\label{thm:safeExpr}
  Let $(\Phi,\lstinline{Prg},\CK)$ be an SPL such that, both $p$ is a valid product and
  let $\mathcal{D}^{\lstinline?Prg?}$ be the derivation proving
  $\phi\deriv \lstinline?Prg?  \; \textsc{ok}$. 
  If $p\models \theta$ and $\theta; \Delta \deriv \lstinline?e?  : \lstinline?T?$ occurs in $\mathcal{D}^{\lstinline?Prg?}$ 
  then $\lsem\Delta\rsem_p\deriv \llangle\lstinline?e?\rrangle_p : \lstinline?T?$.
  Moreover, if \lstinline?struct s? is used in the derivation concluding  $\lsem\Delta\rsem_p\deriv \llangle\lstinline?e?\rrangle_p   : \lstinline?T?$ then $p\models\CK(\lstinline?s?)$.
\end{thm}
\begin{proof}
  If $\mathcal{D}$ be the sub-derivation proving  $\theta; \Delta \deriv \lstinline?e?  : \lstinline?T?$ then, the proof follows by induction on the last rule used in $\mathcal{D}$.
  (Indeed, the proof of   $\lsem\Delta\rsem_p\vdash \llangle\lstinline?e?\rrangle_p\vdash   : \lstinline?T?$  can be obtained from $\theta; \Delta \deriv \lstinline?e?  : \lstinline?T?$
  by removing the sub-derivation about the sub-expression which are removed from the VG and, then, erasing the annotating stuff from the remaining derivation.)
\begin{itemize}
\item Cases \textsc{FT-int}, \textsc{FT-null}, \textsc{FT-malloc} can be transformed in typing derivation
  for the variant $\lsem \lstinline?Prg? \rsem_p$, respectively by using rules \textsc{T-int}, \textsc{T-null}, \textsc{T-malloc}.
\item If the last rule is \textsc{FT-par} then, its premises are $\lstinline?x?:\lstinline?T?\overset{\text{ with } \psi}{\;}  \in\Delta$
  and $\theta \models \psi$. It is easy to see that $\Delta$ is predisposed by the (unique) rule \textsc{FT-fun} in the ascendants of   $\mathcal{D}$ in $\mathcal{D}^{\lstinline?Prg?}$
  and, then,  no typing judgment for expressions modify it. The rule \textsc{FT-fun} put in $\Delta$, for all formal parameter of \lstinline?f? a pair variable-type together  its relative annotation.
  So that,  the proof follows because $p\models \theta$.
\item Consider the case \textsc{FT-app} and let \lstinline?f? be the applied function.
  The premise  $\theta \models\CK(\lstinline!Prg!(\lstinline?f?))$ and Lemma~\ref{lem:funGen}.i ensures that
  the definition of \lstinline?f? is included in the variant $\lsem\lstinline?Prg? \rsem_p$,  as required by premises of $\textsc{T-app}$.
  Lemma~\ref{lem:global} ensures that, in the SPL code base \lstinline?Prg?, the number 
  of formal parameter of \lstinline?f? and of arguments supplied in the call is always the same,
  thus the premise $\theta \models\CK( \overline{\lstinline{e}}) \Leftrightarrow \CK(\overline{\lstinline?T?} \; \overline{\lstinline?x?})$
  is applied to lists of the same length. This latter condition and  Lemma~\ref{lem:funGen}.ii ensure that, in the variant $\llangle\lstinline?e?\rrangle_p$,
  the number of (not removed) formal parameters of  \lstinline?f? and the number of (not removed) parameters supplied to it, match exactly.
  This match is required implicitly by \textsc{T-app}. So, we can easily conclude by induction.
\item Case  \textsc{FT-member} follows by induction, by observing that $\theta \models \CK(\lstinline?T m?)$ ensures that \lstinline?T m? is not removed from  $\llangle\lstinline?e?\rrangle_p$,
  as required by \textsc{T-member} (in the typing derivation for $\lsem \lstinline?Prg? \rsem_p$).
\item Case \textsc{FT-assign}  is similar to the previous one, by replacing it with a \textsc{T-assign} rule in the derivation for $\lsem \lstinline?Prg? \rsem_p$.
\item Cases \textsc{FT-cond},   \textsc{FT-uop}, \textsc{FT-bop} and \textsc{FT-mfree} follow by induction.
\item Case \textsc{FT-seq}. The grammar of the SPL ensures that each expression-list contains at least an expression,
  but is the premise $\theta\models  \exists(\lstinline!e!_1,\ldots,\lstinline!e!_n)$ that ensures that in no variant all expressions are removed from the list.
  Moreover, the premise $\text{ $\forall i\in\{1....,n\}$, if } (\lstinline!T!_i \neq \lstinline?T?) \text{ then } (\theta\models \texttt{neverLast}(i,\tilde{\lstinline!e!}))$
  ensure that expressions ending the list expression in a variant, are always typed with the type of last expression in the CB of the SPL.
  Whenever $p\models \theta \&\&\CK(\lstinline!e!_i)$ holds, by induction on the premise $\theta \&\&\CK(\lstinline!e!_i);\Delta \deriv \lstinline!e!_i :\lstinline!T!_i$
  we have that $\lsem\Delta\rsem_p \deriv \llangle\lstinline!e!_i\rrangle_p :\lstinline!T!_i$. Thus we conclude.
\end{itemize}
\end{proof}

\FDaddCR{We can now prove the type safety result.}\\


\begin{thm}[Type safety for CLC]\label{the:CLC-type-safety}
Let $(\Phi,\lstinline{Prg},\CK)$ be an SPL  with $\Phi=(\mathcal{F},\phi)$ such that  $p\models \phi$, viz. $p$ is a valid product.
If $\phi\deriv \lstinline?Prg?  \; \textsc{ok}$ then  $\vdash \lsem \lstinline?Prg? \rsem_p  \; \textsc{ok}$,
 viz the variant identified by $p$ is well-typed.
\end{thm}
\begin{proof}
  Let $\mathcal{D}^{\lstinline?Prg?}$ be the derivation proving $\phi\deriv \lstinline?Prg?  \; \textsc{ok}$. This derivation  ends with the rule  \textsc{FT-prg}.
We have to verify that we can build the derivation for the variant $\lsem \lstinline?Prg? \rsem_p$   that ends with the rule \textsc{T-prg}.
  \begin{itemize}
  \item The premise $ \lstinline!Prg! \; \textsc{sane}$ does not immediately imply the sanity of $\llangle \lstinline?Prg? \rrangle_p$,
    because the first is about the CB of the whole SPL, while the second is about a specific variant where some code fragments has been removed.
    Sanity requires 3 conditions. The first requires that for every struct name  \lstinline?s? appearing everywhere in \lstinline!Prg! we have  $\lstinline?s?\in\domof(\overline{\lstinline!SD!})$.
    It follows by Lemmas~\ref{lem:annotSubtyping}, \ref{lem:annotStruct}.iii and \ref{thm:safeExpr}.
    The safety for appearing function follows immediately, because the unique typing rule that consider a defined function \lstinline?f? (out of its definition)
    is \textsc{FT-app},
    that includes  $\theta \models\CK(\lstinline!Prg!(\lstinline?f?))$ in its premises: it ensures that if the application is not removed from a variant
    (identified by $p$)
    then also the definition of  \lstinline?f? is included. The presence of \lstinline?f? is ensured by the premise of \textsc{FT-prg}.
  \item The remaining further premises of \textsc{T-prg} are about the function definitions in $\lsem \lstinline?Prg? \rsem_p$.
    The VG definition ensures that if the definition of \lstinline?f? is included in the variant then $p\models \CK(\lstinline!Prg!(\lstinline?f?))$,
so for each function included in $\lsem \lstinline?Prg? \rsem_p$, it holds $p\models \phi \,\&\&\, \CK(\lstinline{Prg}(\lstinline?f?))$.
    The rule \textsc{FT-prg} includes a premise$\ \phi \,\&\&\, \CK(\lstinline{Prg}(\lstinline?f?)) \deriv \lstinline{Prg}(\lstinline?f?)  \; \textsc{ok}\ $
    for each function  $\lstinline{Prg}(\lstinline?f?)=\lstinline?T?\lstinline?f(T?_1\lstinline?x?_1\ldots\lstinline?T?_n\lstinline?x?_n\lstinline?){return e;}?$
    having a definition in the SPL.
 That  premise is, in turn, is the conclusion of a rule \textsc{FT-fun} that contains sufficient premises to can use the rule \textsc{T-fun} for the included function.
In particular, the Lemma~\ref{thm:safeExpr} ensures that we can build the typing derivation for the part (not removed) from expression in the body of \lstinline?f?.
  \end{itemize}
\end{proof}
 


\section{Conclusion and Future Work}\label{sec:conclusion}
In this paper 
we presented: (1)  the syntax an typing of LC, a core calculus for C programs; (2) the syntax and generation mapping of CLC, a core calculus for SPLs of (L)C; and (3)
a type system for CLC that formalizes  a family-based type checking which guarantees that 
all the variants of a well-typed CLC SPL are well typed (L)C programs.

In future work we plan to provide a small-step operational semantics for LC,  prove its type soundness and to formalize other analyses of C programs and family-based analyises of SPLs of C programs (like, e.g., heap safety).

Moreover, we would  like to extend CLC to a larger subset of C and to equip it with a tool chain for analyses of C programs and family-based analyses of SPLs of C programs. 
Analyzing C code is notoriously difficult.
Therefore, instead of aiming at capturing  the whole C syntax and at providing a tool chain for analyzing existing SPLs of C programs, we aim at providing support for:
\begin{itemize}
\item
developing new SPLs of programs written in a subset of C; and
\item
refactoring existing SPLs of C programs into SPLs written in a subset of C that is supported by the tool chain.
\end{itemize}





\bibliographystyle{eptcs}
\bibliography{mybiblio}


\end{document}